\documentclass[onecolumn,12pt,draftclsnofoot]{IEEEtran}
\IEEEoverridecommandlockouts

\usepackage[T1]{fontenc}

\ifCLASSINFOpdf
\else
\fi
\usepackage{caption}
\usepackage{subcaption}
\usepackage{cite}
\usepackage{amsmath,amssymb,amsfonts}
\usepackage{algorithmic}
\usepackage{hyperref}
\hypersetup{bookmarksopen}
\usepackage{booktabs}
\usepackage{graphicx}
\usepackage{textcomp}
\usepackage{array,tabularx, makecell, multirow}
\usepackage{xcolor}
\usepackage{diagbox,multirow,lineno}
\usepackage{amsthm}
{
	
	\newtheorem{corollary}{Corollary}

	\newtheorem{lemma}{Lemma}

}
\newtheorem*{remark}{Remark}
\usepackage[cmintegrals]{newtxmath}
\begin{document}
%
\title{Age of Information in Reservation Multi-Access Networks with Stochastic Arrivals: Analysis and Optimization}
%
%
%

\author{\IEEEauthorblockN{Qian Wang and He (Henry) Chen}
	  \thanks{This work was supported in part by the Innovation and Technology Fund (Project No. ITS/204/20) established under the Innovation and Technology Commission of the Hong Kong Special Administrative Region, China and the CUHK direct grant for research under Project 4055166. The work of Q. Wang is supported in part by the Research Talent Hub PiH/456/21 under Project ITS/204/20.}
	 	\thanks{Q. Wang and H. Chen are with Department of Information Engineering, The Chinese University of Hong Kong, Hong Kong SAR, China (email: \{qwang, he.chen\}@ie.cuhk.edu.hk).}
	 		\thanks{Part of the paper was presented on IEEE ISIT 2022 \cite{isit2022}.}}
\maketitle
\begin{abstract}	
This paper analyzes and optimizes the average Age of Information (AAoI) of Frame Slotted ALOHA with Reservation and Data slots (FSA-RD) in a multi-access network, where multiple users transmit their randomly generated status updates to a common access point in a framed manner. Each frame consists of one reservation slot and several data slots. The reservation slot is further split into several mini-slots. In each reservation slot, users that want to transmit a status update will randomly send short reservation packets in one of the mini-slots to contend for data slots of the current frame. The reservation is successful only if one reservation packet is sent in a mini-slot. The data slots are then allocated to those users that succeed in the reservation slot. 
In the considered FSA-RD scheme, one user with a status update for transmission, termed active user, may need to perform multiple reservation attempts before successfully delivering it. As such, the number of active user(s) in different frames are dependent and thus the probability of making a successful reservation varies from frame to frame, making the AAoI analysis non-trivial. We manage to derive an analytical expression of AAoI for FSA-RD by characterizing the evolution of the number of active user(s) in each frame as a discrete-time Markov chain. We then consider the FSA-RD scheme with one reservation attempt per status update, termed FSA-RD-One. Thanks to the independent frame behaviors of FSA-RD-One, we attain a closed-form expression for its AAoI, which is further used to find the near-optimal reservation probability. Our analysis reveals the impact of key protocol parameters, such as frame size and reservation probability, on the AAoI. Simulation results validate our analysis and show that the optimized FSA-RD outperforms the optimized slotted ALOHA.

\end{abstract}


%

\section{Introduction}
Recent years have witnessed increasing research interest in a new performance metric, Age of Information (AoI), thanks to its capability of quantifying the timeliness of data transmissions in status update systems \cite{kaul2012real,addsun,adds2}. The timeliness of data transmission plays an important role in various Internet of Thing (IoT) applications, particularly in real-time monitoring systems. In these systems, the dynamics of the monitored process should be grasped at the monitor in a timely manner to ensure in-time responses. AoI is defined as the time elapsed since the generation time of the latest received status update at the receiver\cite{kaul2012real}. According to its definition, AoI is jointly determined by the transmission interval, transmission delay, and status generation process.

The trend of massive connectivity of IoT networks \cite{massiot} and the importance of information freshness have recently attracted considerable efforts in optimizing information freshness of multi-access networks. In this context, how to dynamically schedule the data transmission of users in multi-access networks to minimize the network-wide average AoI becomes a critical problem. Polling-based centralized scheduling schemes \cite{kadota2018scheduling,kadota2019minimizing,8437712,bedewy2019age,nikolas1,ass2,ass3} and contention-based random access schemes \cite{unma,slotaloha,gurandom,xrchen,9488897,maatouk2019minimizing1,modernra,rand1} are two major research branches. {The polling-based centralized scheme involves a centralized controller that possesses knowledge of the entire system state, enabling it to make scheduling decisions and send polling packets to grant scheduled devices to conduct data transmissions. On the other hand, the contention-based random access scheme allows all users to contend for access to the shared communication medium, where only the users that succeed in contention can update their statuses.} The sporadic IoT traffic makes contention-based random access schemes more preferable in large-scale networks. This is because centralized scheduling is usually associated with excessive overhead and high operation complexity. In contrast, contention-based multi-access schemes have acceptable overhead with simple operation, and can flexibly adapt to the networks with a varying number of devices.

Previous studies on contention-based multi-access schemes have explored the average AoI of ALOHA and CSMA protocols. The average AoI of slotted ALOHA systems was first characterized in \cite{unma}. Specifically, the ALOHA-alike policy, in which each user transmits its status updates with a fixed transmission probability was analyzed and compared with the centralized scheduling policy \cite{unma}. Inspired by this seminal work, an age-dependent slotted ALOHA policy was devised and analyzed in \cite{slotaloha,gurandom}, where each user randomly accesses the shared channel only when its instantaneous AoI exceeds a predetermined threshold. Different from previous studies considering \textit{generate-at-will} status generation model, the authors in \cite{xrchen,9488897,maatouk2019minimizing1} investigated the contention-based multi-access schemes with stochastic arrival of status updates. In \cite{xrchen,9488897}, status generation at each user is modeled by independent and identically distributed (i.i.d.) Bernoulli process. The analytical expression of the average AoI of a stationary randomized policy was derived, and the asymptotic optimality of slotted ALOHA with a small status arrival rate in the regime of infinite users was analyzed in \cite{xrchen}. An age-based thinning method was proposed to further improve the AoI performance of slotted ALOHA systems. The authors in \cite{9488897} derived approximate expressions for the average AoI of both slotted ALOHA and CSMA schemes by developing a discrete-time model. The AoI of CSMA was also studied in \cite{maatouk2019minimizing1}, where the stochastic hybrid system was applied to derive the accurate average AoI expression for \textit{generate-at-will} model and a tight upper bound of the average AoI for stochastic arrival model, respectively. Very recently, the AoI of irregular repetition slotted ALOHA with successive interference cancellation (SIC) was studied and compared with that of slotted ALOHA in \cite{modernra}. Moreover, SIC-aided age-independent random access was studied and analyzed in \cite{rand1}.

{Apart from the polling-based and contention-based multi-access schemes, there is another type of dynamic allocation schemes, called reservation-based multi-access (R-MA) \cite{Bing2002}, {which allows users to access the network asynchronously or transmit data intermittently. This access mechanism has practical applications in satellite communication \cite{satellite}, Long Term Evolution (LTE) \cite{petar1,petar2}, and WiFi networks \cite{WLAN}.} Specifically, R-MA schemes divide each frame into reservation phase and data transmission phase, which are used to transmit short reservation packets and data packets, respectively. A representative scheme of R-MA is Frame Slotted ALOHA with Reservation and Data slots (FSA-RD) \cite{roberts1973dynamic,szpankowski1983analysis,CASARESGINER201915}. In each frame of FSA-RD, one reservation slot, consisting of multiple mini-slots, is introduced for users to contend for the subsequent data slots. Only those users made successful reservations are allowed to transmit in their reserved data slots of the said frame, leading to non-conflicting data transmission. Furthermore, FSA-RD schemes squeeze the potential collisions among users into the shorter mini-slots, reducing the time overhead for contention. In light of these features, the FSA-RD scheme has great potential to reduce the network-wide AoI. Nevertheless, to the best of the authors' knowledge, the AoI of FSA-RD protocols has not been thoroughly characterized in the open literature.}
	

As an attempt to fill the gap, in this paper we investigate the average AoI (AAoI) of FSA-RD. Specifically, we consider a symmetric multi-access network, where each user transmits its randomly generated status updates to an access point (AP) in a framed manner. Each frame consists of one reservation slot and several data slots. The reservation slot is split into several mini-slots. At the beginning of each frame, users with a status update randomly decide whether to make reservations. For these users that decided to transmit the status update, they will uniformly select one of the mini-slots to send a reservation packet. The reservation is successful only if one reservation packet is sent in a mini-slot. These users that succeed in making reservations will be allocated with a dedicated data slot to update their status. The stochastic arrivals of status updates, the randomness of reservation attempts, and the entangled reservation outcomes altogether make the theoretical analysis of the AAoI for the considered system non-trivial. 

\subsection{Contributions}
{The main contributions of this paper are twofold. Specifically, we first analyze the AAoI of FSA-RD.
In this paper, we refer to these users that have a status update for transmission as active users. Under FSA-RD, active users will randomly decide to make reservations at the beginning of a frame. {Since the reservations made in one frame may fail due to collisions, active users are likely to make multiple reservations. Thus, the number of active users in the next frame will depend on the number of active users in the current frame.} Based on this observation,
we apply a discrete-time Markov chain (DTMC)  to characterize the evolution of the number of active users in each frame, which yields the steady-state distribution for the number of active users in each frame. We then make use of the steady-state distribution to derive an analytical expression of the AAoI of the considered system. 
Nevertheless, the steady-state distribution of the formulated DTMC in FSA-RD can only be calculated numerically due to its complicated evolution. As such, though the derived AAoI expression can be used to design key system parameters (e.g., reservation probability and frame length) of FSA-RD via exhaustive search, {we cannot theoretically observe any properties of the optimized system parameters.}

{To gain more design insights, we then consider a simplified version of FSA-RD, termed FSA-RD-One.
Under FSA-RD-One, each status update is allowed to have at most one reservation/transmission attempt. This makes the distribution of active users in each frame depend only on the status arrival process. And thus the number of active users in different frames are independent. Leveraging such a frame-independent property, we derive a closed-form expression for the AAoI of the FSA-RD-One. After observing the structure of the AAoI expression of FSA-RD-One, we then attain a tight upper bound of the AAoI for FSA-RD-One, which is less complicated and can be used to optimize the reservation probability and frame size efficiently. Based on the AAoI upper bound, a near-optimal reservation probability is derived. {The explicit formula for the near-optimal reservation probability of FSA-RD-One provides more system design insights compared with the analysis of FSA-RD. Specifically,} we find that for a given frame size, the reservation probability that makes the expected number of users making reservations approach the number of mini-slots achieves near-optimal AAoI performance in FSA-RD-One. 
We note that when the frame size is fixed, compared with FSA-RD-One, FSA-RD tends to have more active users in each frame, thus requiring a smaller reservation probability to avoid high reservation collisions. 
}

Simulation results are provided to validate the theoretical expressions of AAoI for the considered two FSA-RD schemes. The impacts of frame size and reservation probability on AAoI are also evaluated. The accuracy of the derived near-optimal reservation probability for the FSA-RD-One is confirmed by simulations. The optimal performance of the two FSA-RD schemes are also compared. The results show that the optimized AAoI of FSA-RD is upper bounded by that of FSA-RD-One. Further, their performance tends to coincide as the status arrival rate of the network increases. Finally, we compare the optimized AAoI of these two FSA-RD-based transmission schemes and that of slotted ALOHA. The comparison results show that the optimized FSA-RD scheme outperforms that of the slotted ALOHA and the performance gap between them enlarges as the status arrival rate of the network goes large. }

\section{System Model}\label{sec2}
\subsection{System Description}
We consider a multi-access network, where $N$ users share a wireless channel to transmit time-sensitive information to an AP in a framed manner. Users are frame-synchronized. The Frame Slotted ALOHA with Reservation and Data slots (FSA-RD) scheme \cite{roberts1973dynamic,szpankowski1983analysis} is adopted. Specifically, each frame consists of $M$ slots. At the beginning of each time slot, the information source of each user randomly generates a time-stamped status update, which is modeled by a Bernoulli process. The network is considered to be symmetric, and the status generation probabilities for all users are equal and denoted by $\rho$\footnote{For asymmetric networks with different status update generation rates, our analysis becomes applicable if those users with higher status update generation rates proactively drop some status update so that a symmetric network with all users having the same ``effective'' status generation rate is mimicked.}. {Our analytical framework can also be used to analyze the generation-at-will model, which can be regarded as a special case of stochastic status arrival model with $\rho=1$.} {The key variables are summarized in Table \ref{tablen}.} We consider that the status update generated within one frame will be considered for transmission at the beginning of the next frame. That is, the status update scheduled for transmission in the current frame cannot be preempted by a newly generated one. For the sake of optimizing the information freshness, each user only needs to keep the latest status update generated within each frame. 
In this case, whether a user has a new status update to transmit at the beginning of the next frame follows a Bernoulli distribution with the parameter $p=1-(1-\rho)^M$. 
Let $I_n(k)$ denote the indicator that equals $1$ if user $n$ has one status update for transmission at the beginning of frame $k$, and equals $0$ otherwise.
{
\begin{table}[!http]
    \centering
    \caption{Notation Table for Key Variables}
    \begin{tabular}{cc}
        \toprule
        $N$ & Number of users \\
        $M$ & Frame size \\
        $V$ & Number of mini-slots in each reservation slot \\
        $\rho$ & Status generation probability at each time slot\\
        $\gamma$ & Reservation probability of active users \\
          $S_{i-1}$ & Slot-level service time of $(i-1)$th received status update \\
        $X_i$ & Slot-level inter-departure time between the $(i-1)$th and $i$th  status receptions\\
        $Y_i$ & Frame-level inter-departure time between the $(i-1)$th and $i$th status receptions\\
        $p_s$& The probability of successful transmission for active users in FSA-RD \\
        $p_s^o$ & The probability of successful transmission for active users in FSA-RD-One\\
        ${\alpha_{i-1}}$ & The index of the reserved slot for $(i-1)$th received status\\
        $\varphi_{\alpha}$ & The probability of reserving $(\alpha-1)$th data slot in FSA-RD\\
        $\varphi_{\alpha}^o$ & The probability of reserving $(\alpha-1)$th data slot in FSA-RD-One\\
        \bottomrule
    \end{tabular}
    \label{tablen}
\end{table}}
\begin{figure}
	\centering \scalebox{1.0}{\includegraphics{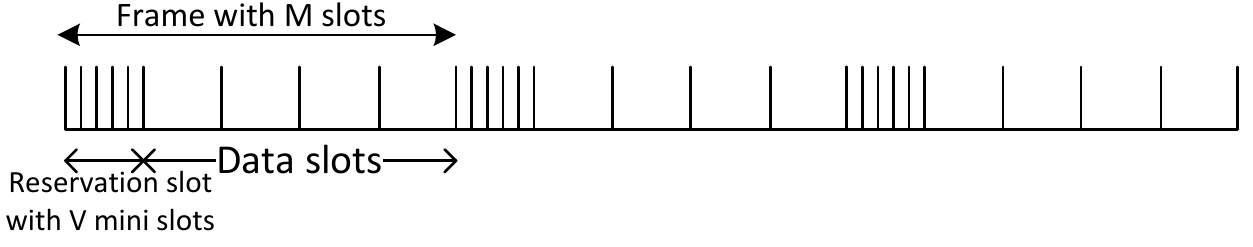}}
	\caption{Frame structure.}
	\label{fig1}
\end{figure}
In FSA-RD, the $M$ slots of each frame are further divided into one reservation slot and $M-1$ data slots. Specifically, the first slot of each frame is the reservation slot, used for making reservations, while the rest $M-1$ data slots are used for sending status updates, as depicted in Fig. \ref{fig1}. Each reservation slot consists of $V$ mini-slots. At the beginning of each frame, each user with a status update for transmission (i.e., active user) will make an ALOHA-alike reservation with the probability $\gamma$. Denote by $J_n(k)$ the reservation indicator of user $n$, which equals to $1$ if user $n$ chooses to make a reservation in frame $k$, and equals to $0$ otherwise. That is, $\mathrm{Pr}(J_n(k)=1|I_n(k)=1)=\gamma$. Once user $n$ decides to reserve, it will uniformly select one of the $V$ mini-slots to send its reservation packet. {The uniform mini-slot selection scheme can reduce the probability of multiple active users selecting the same mini-slot, as compared to non-uniform schemes}\footnote{This study focuses on the evaluation of the AAoI for the considered FSA-RD protocols. Further design and optimization of the reservation scheme with varying reservation probabilities have been left as a future work.}. The transmission of the reservation packet and that of the status update packet are assumed to take $1$ mini-slot and $1$ time slot, respectively, considering that the reservation packet generally contains less information than the status update packet\footnote{{The value of $V$ is determined by the time synchronization precision of the system and the transmission durations of both reservation and status update packets.}}. If more than one user transmits reservation packets in the same mini-slot, a collision occurs and all reservations made in that mini-slot fail; otherwise, the reservation information will be received by the AP. 

At the end of each reservation slot, the AP will inform all users of the reservation results. It is desirable to have a low-overhead feedback scheme for the AP. 
To that end, we note that $V$ mini-slots are used for the reservation, and each user sending its reservation packet knows which mini-slot it tries to reserve. As such, we can let the AP broadcast a $V$-bit feedback, where the $i$th bit equals $1$ if the AP successfully receives a reservation packet in $i$th mini-slot, and equals $0$, otherwise. Once feedback is received, each user will check the bit corresponding to the mini-slot that it sent the reservation packet to. If the bit equals 1 and is one of the first $M-1$ 1's in the feedback sequence, the user knows that its reservation is successful. Otherwise, the user will keep silent. {We note that the allocated data slots for these successful users are indicated by the order of its 1 in all 1's of the feedback sequence.}  
Considering the feedback is short, we ignore the duration of the feedback from the AP for simplicity. Upon the reservation results are released, the users made successful reservations will take turns to transmit their status updates in their allocated data slots. The status update packets are assumed to be successfully received by the AP as there is no collision. Note that at most $M-1$ users can successfully update their statuses in each frame as there are totally $M-1$ data slots. If less than $M-1$ users make successful reservations, the unreserved data slot(s) will be wasted in the current frame\footnote{The investigation for the case with variable frame length has been left as a future work.}. In this sense, the frame size $M\leq V+1$ and should be neither too large nor too small.

{To summarize, the considered system involves two distinct processes for each user: the status generation process and the status transmission process. 
\begin{itemize}
 \item For the status generation process, each user randomly generates a status update at the beginning of each time slot with probability $\rho$. Only the freshest status update generated within one frame will be considered for transmission at the beginning of the next frame.
\item For the status transmission process, at the beginning of a frame, a user with a status update for transmission, will decide whether to make a reservation with probability $\gamma$. Users that successfully reserve a data slot can transmit their status updates to the AP in their allocated data slots. Other users remain silent during the data slots of this frame and prepare for contending a transmission opportunity in the next frame if they have a status update to transmit.  
\end{itemize}}
\subsection{The Evolution of AoI}
The AoI of the user $n$, denoted by $\delta_n(t)$, measures the timeliness of the status updates from the perspective of the AP, which is defined as the time elapsed since the generation time of the most recently received status update from user $n$ at the AP. Mathematically, the AoI $\delta_n(t)$ at time $t$ is $t-u_n(t)$, where $u_n(t)$ denotes the generation time of the latest received status update of user $n$ at the AP until time $t$. 
\begin{figure}
	\centering \scalebox{0.8}{\includegraphics{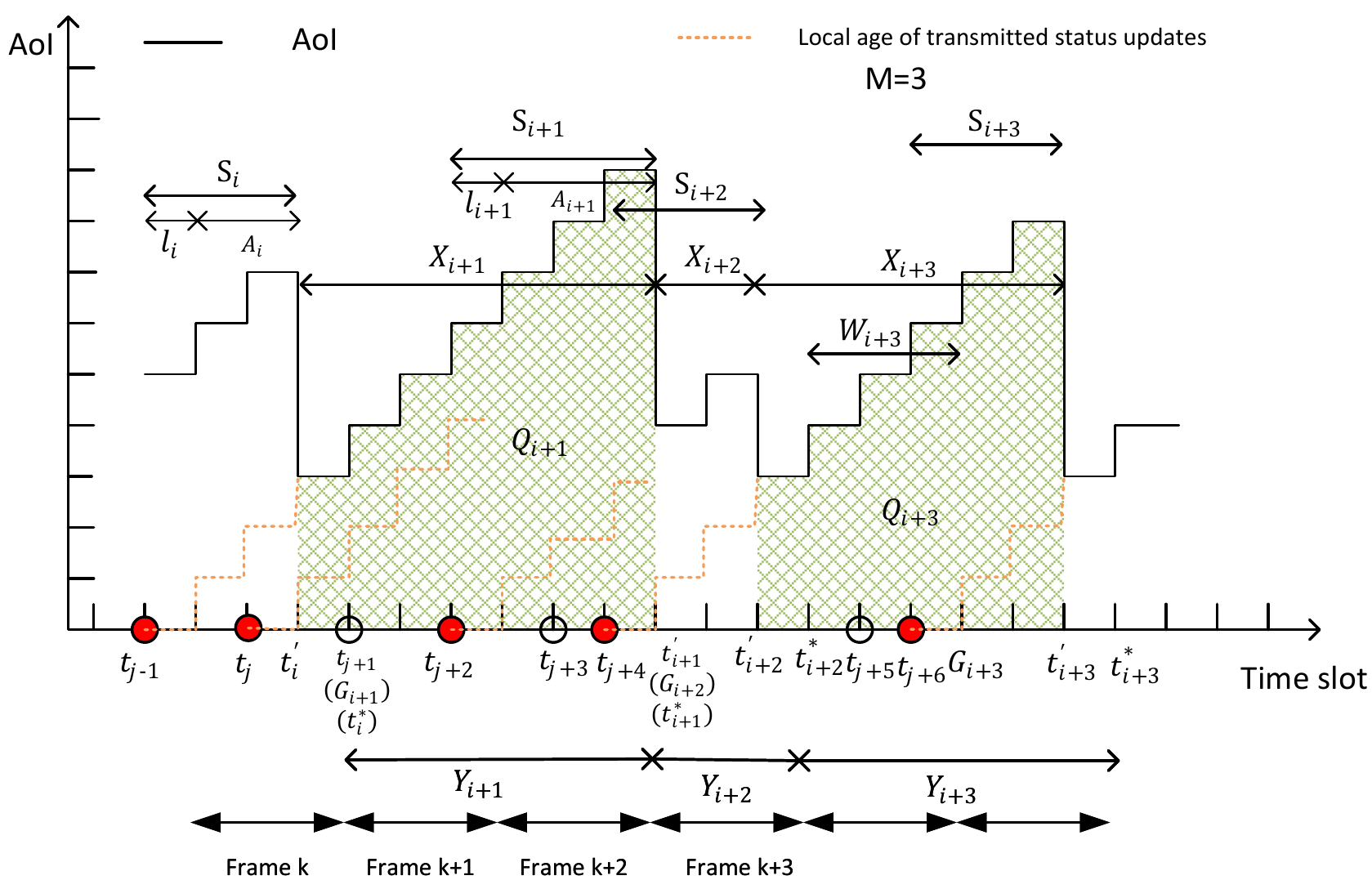}}
	\caption{The evolution of AoI with $M=3$. Each circle indicates the generation of a status update in the corresponding time slot, while the circles in red are those status updates that are considered for transmission. The black solid line is the AoI curve and the orange dotted line is the local age of status updates for transmission.}
	\label{fig2}
\end{figure}

The black solid line in Fig. \ref{fig2} illustrates one example of how the AoI of user $n$ at the AP evolves in the considered system. We can see from Fig. \ref{fig2} that the AoI linearly increases until the AP receives a status update from user $n$, when it is reset to the service time of the successfully transmitted status update. The service time is the difference between the reception time of a new status at the AP and its generation time at user $n$. Denote by $t_j$ the generation time of $j$th status update, and by $t_i^{'}$ the time of the reception of $i$th received status. The different indexes in $t_j$ and $t_i^{'}$ are due to preemption and discard in packet management caused by reservation failures and no-reservation attempts. As shown in Fig. \ref{fig2}, only those status updates that generated at $t_{j-1}$, $t_{j+2}$, $t_{j+4}$ and $t_{j+6}$, were successfully received by the AP at $t_i^{'}$, $t_{i+1}^{'}$, $t_{i+2}^{'}$ and $t_{i+3}^{'}$, respectively. The remaining status updates were replaced by late generated status updates. For example, the status update generated at $t_j$ was transmitted in frame $k+1$, but it might suffer from either reservation failure or no-reservation attempt, and was replaced by the late generated status update. According to the definition of $u(t)$, the service time of $i$th received status update at the AP can be expressed as $S_i=t_i^{'}-u(t_i^{'})$. Recall that each status update may need to wait in the local buffer before being scheduled for transmission. We thus need to define the local age of transmitted status updates of user $n$, denoted by $g_n(t)$, which is the number of time slots that the latest status update has stayed in the local buffer after its generation. In this case, if the AP receives a status update transmitted from user $n$ in time slot $t$, then $\delta_n(t+1)=g_n(t)+1$; otherwise, $\delta_n(t+1)=\delta_n(t)+1$. Other variables shown in Fig. \ref{fig2} will be explained in the following section for further evaluation of AoI.%
\section{Analysis of AAoI for FSA-RD}
 We note that the AoI of each user is identically distributed due to the symmetric network setup. As such, in this section, we focus on analyzing the AAoI of one particular user to represent the network-wide AAoI and omit the user index for brevity. To derive the analytical expression of AAoI, we first present some useful definitions. 
 
 {Based on the frame structure of FSA-RD, we define $t_{i}^{*}$ as the end of the frame within which the AP receives $i$th status update, which is also the beginning of the next frame after the reception of the $i$th status update. As the status updates that randomly arrive in one frame will be considered for transmission in the subsequent frames, we define the waiting time $W_i$ as the time elapsed from $t_{i-1}^{*}$, until the beginning of the first frame when the user has a fresh status update for transmission, denoted by $G_{i}$. 
We thus have $W_i=G_i-t_{i-1}^{*}$. Owing to the frame-based definition, $W_i$ takes values from $\{0,M,2M,...\}$. Recall that whether a user has a new status update to transmit at the beginning of the next frame follows a Bernoulli distribution with
the parameter
$p=1-(1-\rho)^M$. We then can obtain the probability mass function (PMF) of $W_i$, given by $\mathrm{Pr}\{W_i=xM\}=(1-p)^xp$, and we further have 
\begin{equation}\label{eqw}
\mathbb{E}[W_i]=(1-p)M/p,
\end{equation}
\begin{equation}\label{eqw2}
\mathbb{E}[W_i^2]=(p^2-3p+2)M^2/p^2.
\end{equation} We also define $K_i$ as the time interval from $G_i$ until the end of the frame within which the $i$th status update is received at the AP, i.e., $K_i=t_{i}^{*}-G_i$. We further define $Y_i$ as the time interval between the ends of two frames within which the $(i-1)$th and $i$th status updates are received at the AP, a frame-level inter-departure time, i.e., $Y_i=t_{i}^{*}-t_{i-1}^{*}$. Together with the definitions of $W_i$ and $K_i$, we have $Y_i=W_i+K_i$.}

Let $X_i$ denote the inter-departure time of two successive correctly received status updates at the AP, i.e., $X_i=t_{i}^{'}-t_{i-1}^{'}$.
Let $N_t$ denote the number of status updates that have been successfully received by the AP until time $t$. With reference to  \cite{kaul2012real}, the AAoI can be expressed as 
\begin{equation}
\label{eqa11}
\bar{\Delta}=\lim_{t\rightarrow \infty}\frac{N_t}{t}\frac{1}{N_t}\sum_{i=1}^{N_t}Q_i=\frac{\mathbb{E}[Q_i]}{\mathbb{E}[X_i]},
\end{equation} where $Q_i$ is the polygon area as depicted in Fig. \ref{fig2}. The area of $Q_i$ can be further calculated as
\begin{equation}
\label{eqa12}
Q_i=S_{i-1}+S_{i-1}+1+...+S_{i-1}+X_i-1={\left(X_{i}^2-X_{i}\right)}/{2}+S_{i-1}X_i.
\end{equation} By taking the expectations on both sides of \eqref{eqa12} {and substituting the expectations into \eqref{eqa11}, we have
\begin{equation}
\label{eqa13*}
\bar{\Delta}=\frac{\mathbb{E}[X_i^2]}{2 \mathbb{E}[ X_i]}+\frac{\mathbb{E}[S_{i-1}X_i]}{\mathbb{E}[ X_i]}-\frac{1}{2}.
\end{equation} 
To proceed, we need to evaluate the three expectation terms $\mathbb{E}[ X_i]$, $\mathbb{E}[X_i^2]$ and $\mathbb{E}[S_{i-1}X_i]$.

According to the status generation model and transmission model, the service time of $(i-1)$th received status update at the AP, $S_{i-1}$, consists of time intervals that a successfully received status update has experienced in its generation frame since the generation time slot and the transmitting frame(s) since being scheduled for transmission without being replaced until its reception, denoted by $l_{i-1}$ and $A_{i-1}$, respectively. In addition, we can further split $A_{i-1}$ into multiple parts. Specifically, in the transmitting frame(s), once scheduled for transmission, the successfully received status update may experience several frames without being successfully transmitted until the reception frame within which the user reserves a data slot to update its status to the AP. Denote by $z_{i-1}$, the number of frames that the $(i-1)$th received status update has been scheduled for transmission without being replaced by the late generated status updates. Naturally, we have $Z_{i-1}-M$ as the time interval that the status update spent on transmission before its reception frame, where $Z_{i-1}=z_{i-1}M$. Let $\alpha_{i-1}$ denote the time interval the status update spent in the reception frame until its reception time slot. We have $A_{i-1}=Z_{i-1}-M+\alpha_{i-1}$. Thus, $S_{i-1}=l_{i-1}+Z_{i-1}-M+\alpha_{i-1}$. We note that $l_{i-1}$ and $\alpha_{i-1}$ describe slot-level status generation process within the generation frame, and the status transmission process within the reception frame, respectively; while $Z_{i-1}$ characterizes frame-level status transmission. As such,  $l_{i-1}$ and $\alpha_{i-1}$ as well as $Z_{i-1}$ are independent.


Meanwhile, the inter-departure time between the $(i-1)$th and $i$th receptions of status updates consists of three parts: (1) the remaining time slots in the reception frame of the $(i-1)$th received status update since its reception (i.e., $M-\alpha_{i-1}$); (2) the frames without successful status transmission since the $(i-1)$th reception of status updates (i.e., $Y_{i}-M$); and (3) the time slots in the reception frame of the $i$th received status update until its reception (i.e., $\alpha_{i}$). That is, $X_{i}=M-\alpha_{i-1}+Y_{i}-M+\alpha_{i}$. Recall that $\alpha_{i-1}$ and $\alpha_{i}$ are the time intervals that $(i-1)$th and $i$th received status updates spent in their reception frames until being received by the AP, {and thus they are i.i.d.} Thus, we have $\mathbb{E}[X_i]=\mathbb{E}[Y_i]$.

In addition, $Y_i$ characterizes frame-level status transmission process while $\alpha_{i}$ and $\alpha_{i-1}$ focus on slot-level status transmission process in its reception frame. As such, the frame interval(s) without successful status transmission $Y_{i}-M$ is also independent of $\alpha_{i-1}$ and $\alpha_{i}$. Together with the i.i.d. property of the sequence of random variables $\{\alpha_i\}$, we have the following lemma,
\begin{lemma}\label{l1}
	$\mathbb{E}[X_i^2]=\mathbb{E}[Y_i^2]+2\mathrm{Var}(\alpha)$, where $\mathrm{Var}(\alpha)$ is the variance of random variable $\alpha_i$ ($\alpha_{i-1}$ also). 
\end{lemma}
\begin{proof}
See Appendix \ref{pnl1}.
\end{proof}
We now evaluate the term $\mathbb{E}[S_{i-1}X_i]$. Note that $l_{i-1}$ characterizes slot-level status generation process within the generation frame. Thus, $Y_{i}-M$ and $l_{i-1}$ are also independent. Moreover, $Z_{i-1}$ focuses on the transmission process of the $(i-1)$th received status update, while $Y_i$ is related to that of the $i$th received status update. The transmissions of $(i-1)$th and $i$th received status updates are also independent. As such, $Z_{i-1}$ and $Y_i$ are independent. We then have
\begin{lemma}\label{l2}
	$\mathbb{E}[S_{i-1}X_i]=\mathbb{E}[S_{i-1}]\mathbb{E}[Y_i]-\mathrm{Var}(\alpha)$.
\end{lemma}
\begin{proof}
See Appendix \ref{pnl2}.
\end{proof}


By substituting the results in Lemmas \ref{l1} and \ref{l2} into \eqref{eqa13*}, together with $\mathbb{E}[X_i]=\mathbb{E}[Y_i]$, we have 
\begin{equation}
\label{eqa13}
\bar{\Delta}=\frac{\mathbb{E}[Y_i^2]}{2 \mathbb{E}[ Y_i]}+\mathbb{E}[S_{i-1}]-\frac{1}{2}.
\end{equation} Thanks to the i.i.d. property of the two sequences of random variables $\{S_i\}$ and $\{Y_i\}$, we hereafter omit the subscripts of $S_{i-1}$ and $Y_i$, and analyze $\mathbb{E}[S]$, $\mathbb{E}[Y]$ and $\mathbb{E}[Y^2]$ for brevity.}
\subsection{Evaluation of $\mathbb{E}[S]$}\label{seces}
We first calculate the expected service time $\mathbb{E}[S]$. Service time only counts those successfully transmitted status updates. To proceed, we analyze the successful status update probability for user $n$, denoted by $p_s$, once it decides to make a reservation in a frame. Let $I_a^n(k)=1$ denote the successful reception of a status update from user $n$ at the AP within frame $k$. Then we have $p_s=\mathrm{Pr}\left(I_a^n=1|J_n=1, I_n=1\right)$, where we omit the frame indexes of $J_n$, $I_n$ and $I_a^n$, as we consider the converged scenario. Note that the number of users making reservations will affect the successful reservation probability, and the number of users making reservations depends on the number of active users at the beginning of each frame. Let $\mathrm{B}^j_i(\gamma)$ denote the probability that $j$ out of $i$ active users make reservations. Mathematically, we have
\begin{equation}\label{eqr}
\mathrm{B}^j_i(\gamma)=C_i^j(1-\gamma)^{(i-j)}\gamma^j,
\end{equation} according to the considered ALOHA-alike reservation model, where $C_i^j$ is the binomial coefficient. To calculate $p_s$, we need to evaluate the number of active users. 

We note that the number of active users at the beginning of each frame exhibits a Markov property. That is, the number of active users in the current frame depends on how many active users successfully transmit status updates and how many users have new update generation in the previous frame. 
Considering this Markov property, we adopt a discrete-time Markov chain (DTMC) as in \cite{CASARESGINER201915} to model the evolution of the number of active users at the beginning of each frame. The state of the DTMC is the number of active users. Let $P_{ij}$ denote the state transition probability, representing the probability of the system seeing $j$ active users in the next frame given $i$ active users in the current frame. We subsequently analyze the state transition probability to evaluate the steady state of the number of active users.

To proceed, we first calculate the probability that $s$ out of $i$ active users successfully transmit status updates to the AP. Note that the number of users that succeed in reserving a data slot depends on the number of users making reservations. We assume that $j$ out of $i$ active users make reservations, the probability of which is given in \eqref{eqr}. Let $\mathrm{R}^s_j$ denote the probability that among $j$ contending users, $s$ of them succeed in sending their reservation packets in one of the $V$ mini-slots. Recall that when making reservations, each user uniformly selects one of the $V$ mini-slots to send its reservation packet. We notice that the probability of a similar event has been analyzed in  \cite{szpankowski1983analysis}. With reference to \cite[Eq. 6]{szpankowski1983analysis}, we have
\begin{equation}\label{eqr1}
\mathrm{R}^s_j(V)=\frac{(-1)^{s}V!j!}{V^{j}s!}\sum_{m=s}^{\min\{V,j\}}\frac{(-1)^m(V-m)^{j-m}}{(m-s)!(V-m)!(j-m)!}.
\end{equation} 

At the same time, only the first $M-1$ users that make successful reservations can update their statuses in their reserved data slots. We then define $\tilde{\mathrm{R}}^s_j$ as the probability that $s$ out of $j$ contending users successfully reserve a data slot for sending their status updates. We then have 
\begin{equation}\label{eqrs}
\tilde{\mathrm{R}}^s_j(V)=\left\{
\begin{aligned}
&\frac{(-1)^{s}V!j!}{V^{j}s!}\sum_{m=s}^{\min\{V,j\}}\frac{(-1)^m(V-m)^{j-m}}{(m-s)!(V-m)!(j-m)!}, & s\leq M-2, \\
&\sum_{w=s}^{V}\frac{(-1)^{w}V!j!}{V^{j}w!}\sum_{m=w}^{\min\{V,j\}}\frac{(-1)^m(V-m)^{j-m}}{(m-w)!(V-m)!(j-m)!}, & s=M-1.\\
\end{aligned}
\right.
\end{equation} 
The rationale behind \eqref{eqrs} is that when less than $M-1$ users update statuses in the data slots of a frame, those users that succeed in the reservation slot are also those that successfully update statuses. As such,  $\tilde{\mathrm{R}}^s_j=\mathrm{R}^s_j$, when $s\leq M-2$. On the other hand, when $M-1$ users update statuses in the data slots of a frame, there might be more than $M-1$ and up to $V$ users that succeed in the reservation slot. Thus, $\tilde{\mathrm{R}}^s_j=\sum_{w=s}^{V}\mathrm{R}^w_j$, when $s=M-1$. Denote by $\mathrm{D}^s_i(\gamma, V)$ the probability that $s$ out of $i$ active users successfully update statuses in their reserved data slots. According to the law of total probability, we have
\begin{equation}
\mathrm{D}^s_i(\gamma, V)=\sum_{j=s}^{i}\mathrm{B}^j_i(\gamma)\tilde{\mathrm{R}}^s_j(V).
\end{equation} where $\mathrm{B}^j_i(\gamma)$ and $\tilde{\mathrm{R}}^s_j(V)$ are given in \eqref{eqr} and \eqref{eqrs}, respectively.

Recall that both the number of users with successful status transmissions and the number of users with new status arrivals will influence the number of active users of the subsequent frame. The state transition probability $\mathrm{P}_{ij}$ thus, can be expressed as
\begin{equation}\label{eqpij}
\mathrm{P}_{ij}=\sum_{s=\max\{0,i-j\}}^{\min\{i,M-1\}}\mathrm{D}^s_i(\gamma,V)C_{N-i+s}^{j-i+s}\left(1-(1-\rho)^M\right)^{j-i+s}\left((1-\rho)^M\right)^{N-j}
\end{equation} where the term $C_{N-i+s}^{j-i+s}\left(1-(1-\rho)^M\right)^{j-i+s}\left((1-\rho)^M\right)^{N-j}$ is the probability that $(j-i+s)$ out of $(N-i+s)$ users have fresh status updates arriving during frame $k$. The rationale is that the state transition from $i$ active users to $j$ active users results from $s$ among $i$ active users succeeding in transmitting their status updates in frame $k$. Note that the fresh status arrival of $(i-s)$ among $i$ active users without successful transmission will not influence the number of active users in frame $(k+1)$. Hence, the status arrivals of $(N-(i-s))$ users are considered in \eqref{eqpij} and $(j-(i-s))$ out of these users should have new status updates in frame $k$ and become active in frame $(k+1)$, constituting $j$ active users. The value of $s$ ranges from $\max\{0,i-j\}$ to $\min\{i,M-1\}$. The upper bound $\min\{i,M-1\}$ is because there are $i$ active users that might transmit their status updates and at most $M-1$ of them will succeed as each frame only has $M-1$ data slots. The lower bound $\max\{0,i-j\}$ considers the causality. That is, if $j<i$, there should be at least $i-j$ active users becoming inactive in frame  $(k+1)$ (i.e., $i-j$ active users successfully update statuses in frame $k$). Otherwise, there is no least requirement on the number of users to successfully update statuses in frame $k$. 

Let $\pi_i$ denote the steady-state probability of $i$ active users at the beginning of a frame. Based on the state transition probability in \eqref{eqpij}, we can obtain the steady-state distribution $\boldsymbol{\pi}=[\pi_0,\pi_1,...,\pi_N]^T$ by solving the linear equations $\boldsymbol{\pi}=\boldsymbol{\pi}\mathbf{ P}$ together with the fact $\sum_{i=0}^{N} \pi_i=1$. We then rely on the steady-state distribution $\boldsymbol{\pi}$ to calculate the successful transmission probability of an active user making reservations, i.e., $p_s$. The following lemma gives the analytical expression of $p_s$,
\begin{lemma}\label{l3}
$p_s=\sum_{n_1=0}^{N-1}\frac{\pi_{n_1+1} (n_1+1)}{\sum_{j=0}^{N-1}\pi_{j+1}(j+1) }\sum_{n_2=0}^{n_1}\mathrm{B}_{n_1}^{n_2}(\gamma)\sum_{n_3=1}^{\min\{V,n_2+1\}}\mathrm{R}^{n_3}_{n_2+1}(V)\frac{\min\{{n_3},M-1\}}{n_2+1}$.
\end{lemma}
\begin{proof}
	See Appendix \ref{pl1}.
\end{proof}

{{The following corollary characterizes the probability that one user successfully transmits its status update in the reserved $(\alpha-1)$th data slot (i.e., $\alpha$th slot of a frame), denoted by $\varphi_\alpha$, where $\alpha\in\{2,...,M\}$ and $\sum_{\alpha=2}^M\varphi_{\alpha}=p_s$. The proof is omitted as it can be directly inferred from Lemma \ref{l3}.
\begin{corollary}\label{co11}
	$\varphi_\alpha=\sum_{n_1=0}^{N-1}\frac{\pi_{(n_1+1)} (n_1+1)}{\sum_{j=0}^{N-1}\pi_{(j+1)}(j+1) }\sum_{n_2=0}^{n_1}\mathrm{B}_{n_1}^{n_2}(\gamma)\sum_{n_3=\alpha-1}^{\min\{V,n_2+1\}}\mathrm{R}^{n_3}_{n_2+1}(V)\frac{1}{n_2+1}$.
\end{corollary}

Recall that $S=l+Z-M+\alpha$, and thus $\mathbb{E}[S]=\mathbb{E}[l]+\mathbb{E}[Z]-M+\mathbb{E}[\alpha]$. 
Due to the independent generation of status updates in each slot, we have $l\in\{1,2,...,M\}$ and the probability that a status update is generated $l$ slots before its transmission frame without preemption is given by $\phi_{l}=\rho(1-\rho)^{l-1}$, which is the product between the probability of one status generation at the $(M-l)$th slot of a frame and the probability of no status generation in the following $l-1$ consecutive slots. We then have
\begin{equation}\label{el}
\mathbb{E}[l]=\frac{\sum_{l=1}^{M}\phi_{l}l}{\sum_{l=1}^{M}\phi_{l}}=\frac{1}{\rho}-\frac{M(1-\rho)^M}{1-(1-\rho)^M}.
\end{equation}
For $\mathbb{E}[Z]$, recall that $Z=zM$, where $z$ is the number of frames that the received status update spent on transmission since its generation frame. We then have $\mathrm{Pr}(z=i)=\left((1-\rho)^M\right)^{i-1}(1-\gamma p_s)^{i-1}\gamma p_s$, where $\left((1-\rho)^M\right)^{i-1}$ represents the probability that no fresh status update is generated to replace (preempt) the current status update during transmission before its reception frame, $(1-\gamma p_s)^{i-1}$ is the probability that the status update is not received by the AP in previous $(i-1)$ frames, and $\gamma p_s$ is the successful probability of the status update being received in the $i$th frame since its generation frame. Thus, we have 
\begin{equation}
\mathbb{E}[Z]=M\mathbb{E}[z]=M\frac{\sum_{i=1}^{\infty}\mathrm{Pr}(z=i)i}{\sum_{i=1}^{\infty}\mathrm{Pr}(z=i)}=\frac{M}{1-(1-\gamma p_s)(1-\rho)^M}.
\end{equation}

As for $\mathbb{E}[\alpha]$, based on Corollary \ref{co11}, we have 
$\mathbb{E}[\alpha]={\sum_{\alpha=2}^{M}\varphi_{\alpha}\alpha}/{p_s}$. In this regard, $\mathbb{E}[S]$ can be expressed as 
\begin{equation}\label{es1}
\mathbb{E}[S]=\frac{1}{\rho}-\frac{M(1-\rho)^M}{1-(1-\rho)^M}+\frac{\sum_{\alpha=2}^{M}\varphi_{\alpha}\alpha}{p_s}+\frac{M}{1-(1-\gamma p_s)(1-\rho)^M}-M.
\end{equation}}
\subsection{Evaluation of $\mathbb{E}[Y]$ and $\mathbb{E}[Y^2]$}\label{eyey2}
We note that $W$ and $K$ are independent because $W$ only depends on status arrival rate $\rho$. Recall that $Y=W+K$, we thus have}
\begin{equation}
\label{aeq10}
\mathbb{E}[ Y]=\mathbb{E}[W]+\mathbb{E}[K],
\end{equation}
\begin{equation}
\label{aeq11}
\mathbb{E}[Y^2]=\mathbb{E}[(W+K)^2]=\mathbb{E}[W^2]+\mathbb{E}[K^2]+2\mathbb{E}[W]\mathbb{E}[K],
\end{equation} where $\mathbb{E}[W]$ and $\mathbb{E}[W^2]$ are given in \eqref{eqw} and \eqref{eqw2}.

As for variable $K$, we note that when there is a status update for transmission at the beginning of one frame, the probability of successfully transmitting the status update within this frame is $\gamma p_s$. Furthermore, if the status update is not received by the AP, the transmission process will continue until a successful status transmission as we allow multiple reservation attempts for each status update. It is worth mentioning that although the fresh status updates generated before the reception frame will replace the stale status update for transmission, the transmission process will not be affected. Thus, the PMF of $K$ can be expressed as $\mathrm{Pr}\{K=xM\}=(1-\gamma p_s)^{x-1}\gamma p_s$, where $x\in\{1,2,...\}$. Thus, we have
\begin{equation}\label{aek}
\mathbb{E}[K]=\frac{M}{\gamma p_s},
\end{equation}
\begin{equation}\label{aek2}
\mathbb{E}[K^2]=\frac{M^2}{(\gamma p_s)^2}+\frac{M^2(1-\gamma p_s)}{(\gamma p_s)^2}.
\end{equation}
Regarding the expected value of $Y^2$, by substituting \eqref{eqw}-\eqref{eqw2} and \eqref{aek}-\eqref{aek2} into \eqref{aeq11}, we have $\mathbb{E}[Y^2]$ as follows,
\begin{equation}
\mathbb{E}[Y^2]=M^2-\frac{3M^2}{\gamma p_s}+\frac{2M^2}{(\gamma p_s)^2}+\frac{2M^2}{\gamma p_s \left(1-(1-\rho)^M\right)}-\frac{3M^2}{1-(1-\rho)^M}+\frac{2M^2}{\left(1-(1-\rho)^M\right)^2}
\end{equation}
Based on $\mathbb{E}[Y^2]$, $\mathbb{E}[Y]$ and $\mathbb{E}[S]$, after some manipulations, we obtain the analytical expression of AAoI for FSA-RD, given by
\begin{equation}\label{agen3}
\bar{\Delta}=\frac{M}{\gamma p_s}-\frac{M}{2}+\frac{1}{\rho}+\frac{\sum_{\alpha=2}^{M}\varphi_{\alpha}\alpha}{p_s}-\frac{1}{2}.
\end{equation}
where $p_s$ and $\varphi_{\alpha}$ are given in Lemma \ref{l3} and Corollary \ref{co11}, respectively. 

For the considered FSA-RD protocol, there are two parameters that need to be designed, reservation probability $\gamma$ and frame size $M$. It is noted that the steady-state distribution of the formulated DTMC in FSA-RD can only be calculated numerically due to its complicated evolution. Thus, with the AAoI expression in \eqref{agen3}, we have to apply the exhaustive search method to optimize $\gamma$ and $M$. {As $\gamma\in(0,1]$ and $M\in\{2,3,..,V+1\}$, the complexity of the exhaustive search is acceptable. Nevertheless, the system design insights offered by such a method are limited. This motivates us to study a simplified version of FSA-RD in the next section.}

\section{FSA-RD with One Reservation Attempt per Status Update}
In this section, we consider a simplified scheme, FSA-RD with one reservation attempt per status update (FSA-RD-One). That is, a newly generated status update can only be transmitted in the subsequent frame\footnote{We note that similar restrictions have been applied in the design of AoI-oriented random access protocols, see e.g., \cite{modernra}.}. More specifically, only the status update most recently generated during a frame will be kept for transmission in the subsequent frame. No matter whether the transmission of the status update succeeds or not, the status update will be discarded at the end of the transmission frame. As such, in any frame $k$, whether a user has a status update to transmit also follows a Bernoulli distribution. The corresponding probability equals the probability that there is at least one status update generated in the previous $M$ slot of frame $k-1$. 

{In the above context, whether a user is active at the beginning of each frame becomes i.i.d., and thus the distribution of the number of active users is now frame-independent.} We then have $\mathrm{Pr}(I_n(k)=1)=1-(1-\rho)^M$, and the number of active users, i.e., $\sum_{n=0}^{N}I_n$ in each frame is i.i.d. 
We remark that when the status update generation probability $\rho$ is high, users under FSA-RD-One will behave similar to that under FSA-RD.  {Specifically, as $\rho$ goes large, $\mathrm{Pr}(I_n(k)=1)$ increases, meaning that users in each frame are very likely to have a newly generated status update in the previous frame for transmission. In this case, users under FSA-RD-One or FSA-RD tend to transmit the same status update.
In this sense, the performance of FSA-RD-One serves as a good  approximation to that of FSA-RD, {especially when the term $1-(1-\rho)^M$ becomes relatively large as $\rho$ increases.}} Also, $z_{i-1}=1$ as each status update has at most one reservation attempt. Hence, $S_{i-1}=l_{i-1}+\alpha_{i-1}$ and Lemmas \ref{l1} and \ref{l2} still hold for this setup. We now focus on three terms in \eqref{eqa13} to derive the closed-form expression of AAoI for this simplified scheme.
\subsection{ Evaluation of $\mathbb{E}[S]$}
As in Sec \ref{seces}, we first analyze the successful transmission probability of an active user when making a reservation under the FSA-RD-One scheme, denoted by $p_s^{o}$. As we consider a symmetric network, $p_s^{o}$ is the same for all users due to the symmetric network considered. Thus, we focus on analyzing $p_s^{o}$ for any arbitrary user $n$. Note that as there is no retransmission, the number of active users depends on the status arrival process only. As such, the number of active users excluding user $n$ with a status update follows a Binomial distribution, and the probability that $n_1$ of $N-1$ users are active can be represented by $\mathrm{B}_{N-1}^{n_1}(1-(1-\rho)^M)$ as in \eqref{eqr}. We then have the following lemma that gives a closed-form expression of $p_s^o$,
\begin{lemma}\label{lem1}
	$p_s^o=\sum_{n_1=0}^{N-1}\sum_{n_2=0}^{n_1}\sum_{n_3=1}^{\min\{V,n_2+1\}}\mathrm{B}_{N-1}^{n_1}(1-(1-\rho)^M)\mathrm{B}_{n_1}^{n_2}(\gamma)\mathrm{R}^{n_3}_{n_2+1}(V)\frac{\min\{{n_3},M-1\}}{n_2+1}$.
\end{lemma}
\begin{proof}
	See Appendix \ref{pl4}.
\end{proof}
We then can derive the corresponding probability that one user successfully transmits its status update in the reserved $(\alpha-1)$th data slot, denoted by $\varphi_{\alpha}^o$, where $\alpha\in\{2,...,M\}$ and $\sum_{\alpha=2}^M\varphi_{\alpha}^o=p_s^o$, given in the following corollary. The proof is also omitted, as it can be directly inferred from Lemma \ref{lem1}.
\begin{corollary}\label{co1}		$\varphi_\alpha^o=\sum_{n_1=0}^{N-1}\sum_{n_2=0}^{n_1}\sum_{n_3=\alpha-1}^{\min\{V,n_2+1\}}\mathrm{B}_{N-1}^{n_1}(1-(1-\rho)^M)\mathrm{B}_{n_1}^{n_2}(\gamma)\sum_{n_3=1}^{\min\{V,n_2+1\}}\mathrm{R}^{n_3}_{n_2+1}(V)\frac{1}{n_2+1}$.
\end{corollary}
Recall that $S=l+\alpha$, and thus $\mathbb{E}[S]=\mathbb{E}[l]+\mathbb{E}[\alpha]$. 
Owing to the same generation process of status updates as for FSA-RD, $\mathbb{E}[l]$ in \eqref{el} still holds. Thus, we have 
	    \begin{equation}\label{es}
		\mathbb{E}[S]=\frac{1}{\rho}-\frac{M(1-\rho)^M}{1-(1-\rho)^M}+\frac{\sum_{\alpha=2}^{M}\varphi_{\alpha}^o\alpha}{p_s^o}.
		\end{equation}
\subsection{Evaluation of $\mathbb{E}[Y]$ and $\mathbb{E}[Y^2]$}
We note that it is not easy to directly calculate the distribution of $K$. Inspired by \cite{gu2019timely,chen2016age}, we apply the recursive method to calculate the expectations of $K$ and $K^2$. Then, we make use of the relationship among $K$, $W$ and $Y$ in \eqref{aeq10} and \eqref{aeq11} to obtain $\mathbb{E}[Y]$ and $\mathbb{E}[Y^2]$. The term $K$ has two different behaviors depending on whether the status update is successfully received by the AP or not: \textbf{1)} If the status update is successfully received in its transmission frame, $K$ equals $M$. This event happens with a probability of $\gamma p_s^o$; \textbf{2)} If not, the user needs to wait for the generation of a new status update as retransmission is not allowed. Then, $K=M+W^{'}+K^{'}$. The corresponding probability is $1-\gamma p_s^o$. Here, $W^{'}$ denotes the waiting time of a new status available for transmission, and $K^{'}$ is the remaining frame(s) to successfully transmit a status update to the AP. Multiple status updates may be generated and discarded during this process. {We notice that $\mathbb{E}[K]=\mathbb{E}[K^{'}]$ and $\mathbb{E}[K^2]=\mathbb{E}[{K^{'}}^2]$  due to the same evolution, and $\mathbb{E}[W]=\mathbb{E}[W^{'}]$ and $\mathbb{E}[W^2]=\mathbb{E}[{W^{'}}^2]$  due to the i.i.d. process. {Besides, $W'$ and $K'$ are independent according to the definition.} As such, $\mathbb{E}[K]$ can be calculated as follows,
\begin{equation}
\label{ek}
\mathbb{E}[K]=\gamma p_s^oM+(1-\gamma p_s^o)(M+\mathbb{E}[W]+\mathbb{E}[K]),
\end{equation}where the first term refers to the first case of successful transmission for the status update and the second term represents the second case that the user waits for a new status update for transmission and starts to transmit the new one. Similarly, $K^2$ equals $M^2$ with probability $\gamma p_s^o$ and equals $(M+W'+K')^2$ with probability $1-\gamma p_s^o$. Thanks to the independence between $W'$ and $K'$, we have  
\begin{equation}\label{mek2}
\mathbb{E}[K^2]=\gamma p_s^oM^2+(1-\gamma p_s^o)\left( M^2 + \mathbb{E}[K^2] + \mathbb{E}[W^2]\right)+
(1-\gamma p_s^o)\left(2\mathbb{E}[K]\mathbb{E}[W]+2M\left(\mathbb{E}[K]+\mathbb{E}[W]\right)\right).
\end{equation}}By substituting $\mathbb{E}[W]$ and $\mathbb{E}[W^2]$ from \eqref{eqw} and \eqref{eqw2} and the above expressions of $\mathbb{E}[K]$ and $\mathbb{E}[K^2]$ into \eqref{aeq10} and \eqref{aeq11}, the expectation of frame-level inter-departure time of two successive correctly transmitted status updates, $\mathbb{E}[Y]$, and the expectation of the square of frame-level inter-departure time, $\mathbb{E}[Y^2]$, can be obtained as follows
\begin{equation}
\label{ey}
\mathbb{E}[Y]=\mathbb{E}[W]+\mathbb{E}[K]=\frac{M}{\gamma p_s^o}+\frac{M(1-\rho)^M}{\gamma p_s^o\left(1-(1-\rho)^M\right)}.
\end{equation} 
\begin{equation}
\mathbb{E}[Y^2] = \frac{M^2}{\gamma p_s^o}+\frac{2M^2\!\!\left( \frac{1}{\gamma p_s^o}-\left(1-(1-\rho)^M\right) \right)}{\gamma p_s^o\left(1-(1-\rho)^M\right)^2}+\frac{M^2(1-\rho)^M}{\gamma p_s^o(1-(1-\rho)^M)}
\end{equation}
Based on $\mathbb{E}[Y^2]$, $\mathbb{E}[Y]$ and $\mathbb{E}[S]$, we obtain a closed-form expression of the AAoI for FSA-RD-One, given by
\begin{equation}\label{age}
\bar{\Delta}^o\!\!=\!\!\frac{M}{\gamma p_s^o\left(1-(1-\rho)^M\right)}-\frac{M(1-\rho)^M}{1-(1-\rho)^M}+\frac{1}{\rho}-\frac{M+1}{2}+\frac{\sum_{\alpha=2}^{M}\varphi_{\alpha}^o\alpha}{p_s^o},
\end{equation} where $p_s^o$ and $\varphi_{\alpha}^o$ are given in Lemma \ref{lem1} and Corollary \ref{co1}, respectively. {Note that when $\rho=1$, the above AAoI expression of FSA-RD-One and that of FSA-RD given in \eqref{agen3} will reduce to the same expression, which verifies the correctness of our analysis and confirms our deduction on the relationship between FSA-RD-One and FSA-RD.}

According to the closed-form expression in \eqref{age}, it is still non-trivial to attain the optimal system parameters. We then think of making an approximation for ease of parameter optimization. Note that $\alpha$ takes value from $\{2,...,M\}$. Thus, we have $\mathbb{E}[\alpha]\leq M$. As such, we have the following upper bound of AAoI,
\begin{equation}\label{ageu}
\bar{\Delta}^U=\frac{M}{\gamma p_s^o\left(1-(1-\rho)^M\right)}-\frac{M(1-\rho)^M}{1-(1-\rho)^M}+\frac{1}{\rho}+\frac{M-1}{2}.
\end{equation} We note that the gap between $\bar{\Delta}^U$ and $\bar{\Delta}^o$ is bounded by $M$.
 
Observing the upper bound expression of AAoI, we can find that given the status arrival rate $\rho$, the number of mini-slots in reservation slot $V$, and a fixed frame size $M$, $\bar{\Delta}^U$ decreases as $\gamma p_s^o$ increases, where $p_s^o$ is the successful transmission probability when making reservations. As such, finding the optimal reservation probability that minimizes the AAoI upper $\bar{\Delta}^U$ is equivalent to finding that maximizes $\gamma p_s^o$, when other system parameters are given. 

Due to the complex expression of $p_s^o$ given in Lemma \ref{lem1}, it is difficult to attain the optimal $\gamma$ by calculating the zero root of the first-order derivative of $p_s^o \gamma$. To proceed, we seek to find a bound for the term $p_s^o \gamma$. Along this line, we realize that $p_s^o \gamma$ is actually the probability that an active user successfully transmits its status update to the AP. Recall that the successful status update of an active user has two conditions: 1) the user sends a reservation packet to one of the $V$ mini-slots and no other users send reservation packets to the same mini-slot; 2) the reserved collision-free mini-slot is one of the first $M-1$ reserved mini-slots without collision. 
Let $\widetilde{p}_s$ denote the probability of reserving a collision-free mini-slot once a user makes a reservation. Then the probability that an active user reserves a collision-free mini-slot can be given by $\widetilde{p}_s\gamma$, which is exactly the probability of the above-mentioned condition 1). In this sense, we can deduce that $p_s^o\gamma \le\widetilde{p}_s\gamma$ (i.e., $p_s^o \le\widetilde{p}_s$). To make it rigorous, we next characterize the relationship between $p_s^o$ and $\widetilde{p}_s$ mathematically. To proceed, we first drive an expression of $\widetilde{p}_s$ given by
\begin{corollary}\label{col3}
    $\widetilde{p}_s{=}\sum_{n_1=0}^{N-1}\sum_{n_2=0}^{n_1}\sum_{n_3=1}^{\min\{V,n_2+1\}}\mathrm{B}_{N-1}^{n_1}(1-(1-\rho)^M)\mathrm{B}_{n_1}^{n_2}(\gamma)\mathrm{R}^{n_3}_{n_2+1}(V)\frac{n_3}{n_2+1}$.
\end{corollary} 
The proof of Corollary \ref{col3} is omitted as it can be directly inferred from Lemma \ref{lem1}. With the formula above, we are now ready to analyze the relationship between $p_s^o\gamma$ and $\widetilde{p}_s\gamma$. By observing the similar expression of $p_s^o$ in Lemma \ref{lem1} and $\widetilde{p}_s$ in Corollary \ref{col3}, we can confirm our earlier deduction that $p_s^o\leq\widetilde{p}_s$, $\forall\gamma$. Specifically, the only difference between $p_s^o$ and $\widetilde{p}_s$ lies in the third summation over $n_3$, which is $\sum_{n_3=1}^{\min\{V,n_2+1\}}\mathrm{R}^{n_3}_{n_2+1}(V)$ $\frac{\min\{n_3, M-1\}}{n_2+1}$ in $p_s^o$ and $\sum_{n_3=1}^{\min\{V,n_2+1\}}\mathrm{R}^{n_3}_{n_2+1}(V)$ $\frac{n_3}{n_2+1}$ in $\widetilde{p}_s$, respectively. Note that $\mathrm{R}^{n_3}_{n_2+1}(V)$ defined in \eqref{eqr1} is non-negative. We then have $p_s^o\le\widetilde{p}_s$ and  $p_s^o\gamma\leq\widetilde{p}_s\gamma$, $\forall\gamma$. 
In this context, we can optimize the term ${p}^o_s\gamma$ indirectly through optimizing its upper bound $\widetilde{p}_s\gamma$. 
This is largely motivated by the fact that we can have a rather concise expression of $\widetilde{p}_s$. Specifically, from the perspective of an arbitrary active user reserving a collision-free mini-slot rather than considering how many active users reserving a collision-free mini-slot, the expression of $\widetilde{p}_s$ can be further simplified as follows
\begin{corollary}\label{co6}
$\widetilde{p}_s=\sum_{n_1=0}^{N-1}\sum_{n_2=0}^{n_1}\mathrm{B}_{N-1}^{n_1}(1-(1-\rho)^M)\mathrm{B}_{n_1}^{n_2}(\gamma)V\frac{1}{V}\left(1-\frac{1}{V}\right)^{n_2}=\left(1-\frac{\gamma \left(1-(1-\rho)^M\right)}{V}\right)^{N-1}$. \end{corollary}
\begin{proof}
    See Appendix \ref{pc4}.
\end{proof}
Thanks to the concise expression of $\widetilde{p}_s$, we can derive a closed-form expression of the optimized reservation probability when other parameters are fixed, given in the following lemma:
\begin{lemma}\label{l5}
	Given status arrival rate $\rho$, the number of mini-slots in reservation slot $V$ and the number of users $N$, and a fixed frame size $M$, the optimized reservation probability $\gamma^* =\min\{1, \frac{V}{N\left(1-(1-\rho)^M\right)}\}$.
\end{lemma}
\begin{proof}
	See Appendix \ref{pl5}.
\end{proof}
The simulation results presented in Sec. \ref{nr} confirm that the optimized reservation probability $\gamma^*$ is very close to optimal.

We are now ready to optimize the frame length $M$. For each given frame length $M$, we can find the corresponding $\gamma^*(M)$ by applying Lemma \ref{l5}. By substituting $\gamma^*(M)$ into \eqref{age}, we can express the AAoI as the function of frame length $M$. We then can optimize the frame length $M$ via an exhaustive search over its feasible region. Recall that $M$ is an integer and $M\in\{2,3,..., V+1\}$ as at most $V$ users can reserve a collision-free mini-slot in each reservation slot. Thus, the exhaustive search can be applied for the effective optimization of frame length $M$. The optimality of this method will be verified in numerical results.
\begin{remark}\label{rem1}
{We can see from Lemma \ref{l5} that the optimized reservation probability $\gamma^*$ tries to push the expected number of users making reservations to approach the number of mini-slots, $V$. This is reasonable as when the number of users making reservations exceeds $V$, the larger the number of users making reservations, the smaller the probability for each user to reserve a collision-free mini-slot, leading to a smaller successful status update probability. When the expected number of users making reservations is smaller than $V$, which is caused by a small reservation probability, it is also less likely to successfully transmit status updates as users may be too conservative to make reservation attempts. 

We also note that compared with FSA-RD, with the same system setups, FSA-RD-One tends to have a smaller number of active users at the beginning of each frame. The rationale is that in FSA-RD, the scheduled status updates that are not successfully received by the AP can be transmitted as long as they are not replaced by the latest generated status updates. But in FSA-RD-One, each status update scheduled for transmission has only one reservation attempt, thus can be transmitted at most once. In this sense, at the beginning of each frame, the number of active users in FSA-RD tends to be larger than that in FSA-RD-One. For the same reservation probability, it is thus more likely for FSA-RD to suffer from more reservation collisions compared with FSA-RD-One. While the optimized reservation probability is to achieve the expected $V$ users sending reservation packets for FSA-RD-One. With the same reservation probability, the number of the expected users sending reservation packets for FSA-RD will be larger than $V$, yielding a higher probability of reservation collision. In this sense, the optimal reservation probability of FSA-RD should be smaller than that of FSA-RD-One, which is later verified via simulation results.} 
\end{remark}
\section{Numerical results}\label{nr}
In this section, both numerical and analytical results of the AAoI of FSA-RD and FSA-RD-One are presented. 
\subsection{Validation of the analysis of AAoI for FSA-RD and FSA-RD-One}
{\color{blue}\begin{figure}[!t]
	\centering
	\begin{subfigure}[b]{0.45\textwidth}
		\centering
		\includegraphics[width=\textwidth]{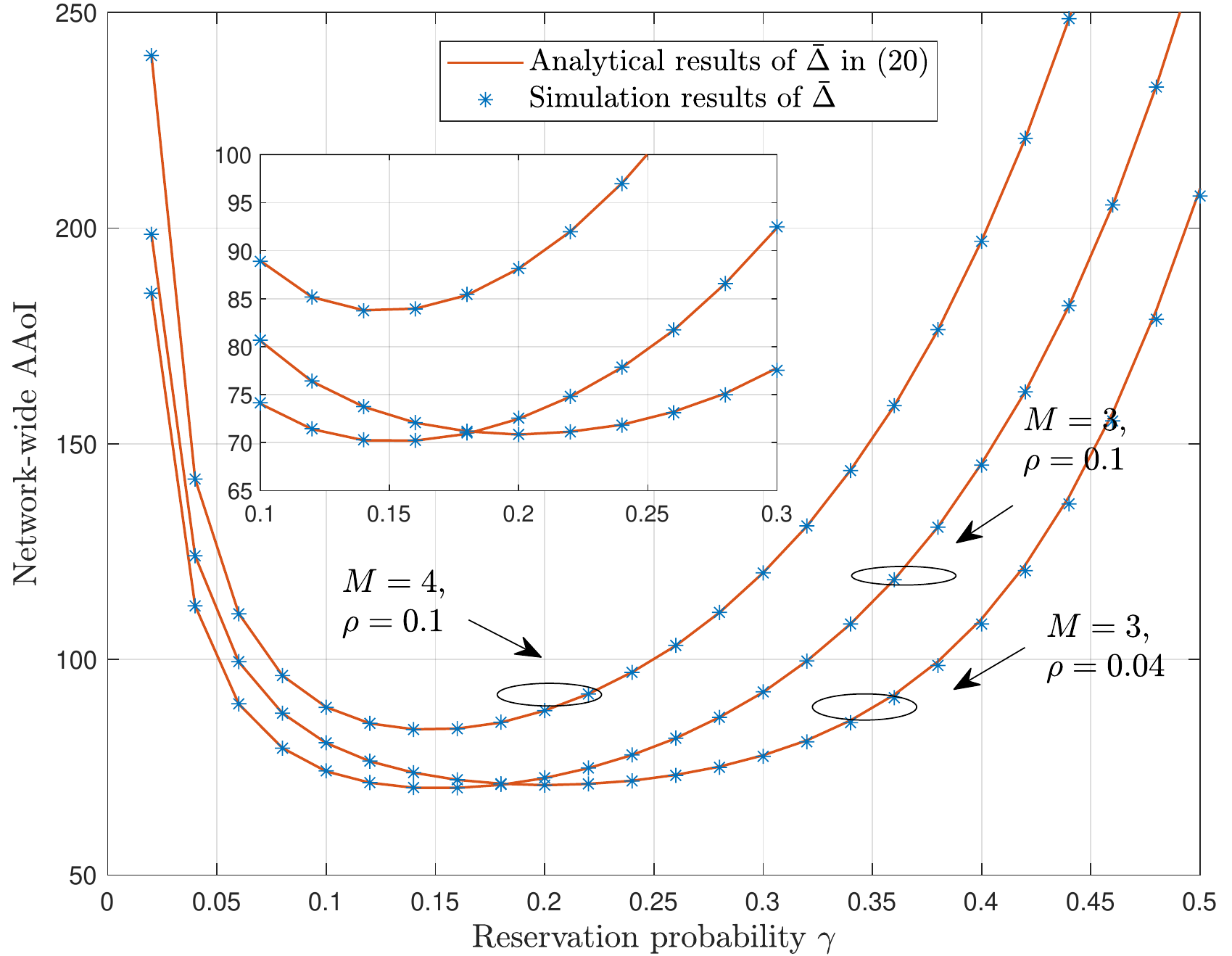}
		\caption{FSA-RD with $N=30$, $V=4$.}
		\label{fig3a}
	\end{subfigure}
	\begin{subfigure}[b]{0.45\textwidth}
		\centering
  \includegraphics[width=\textwidth]{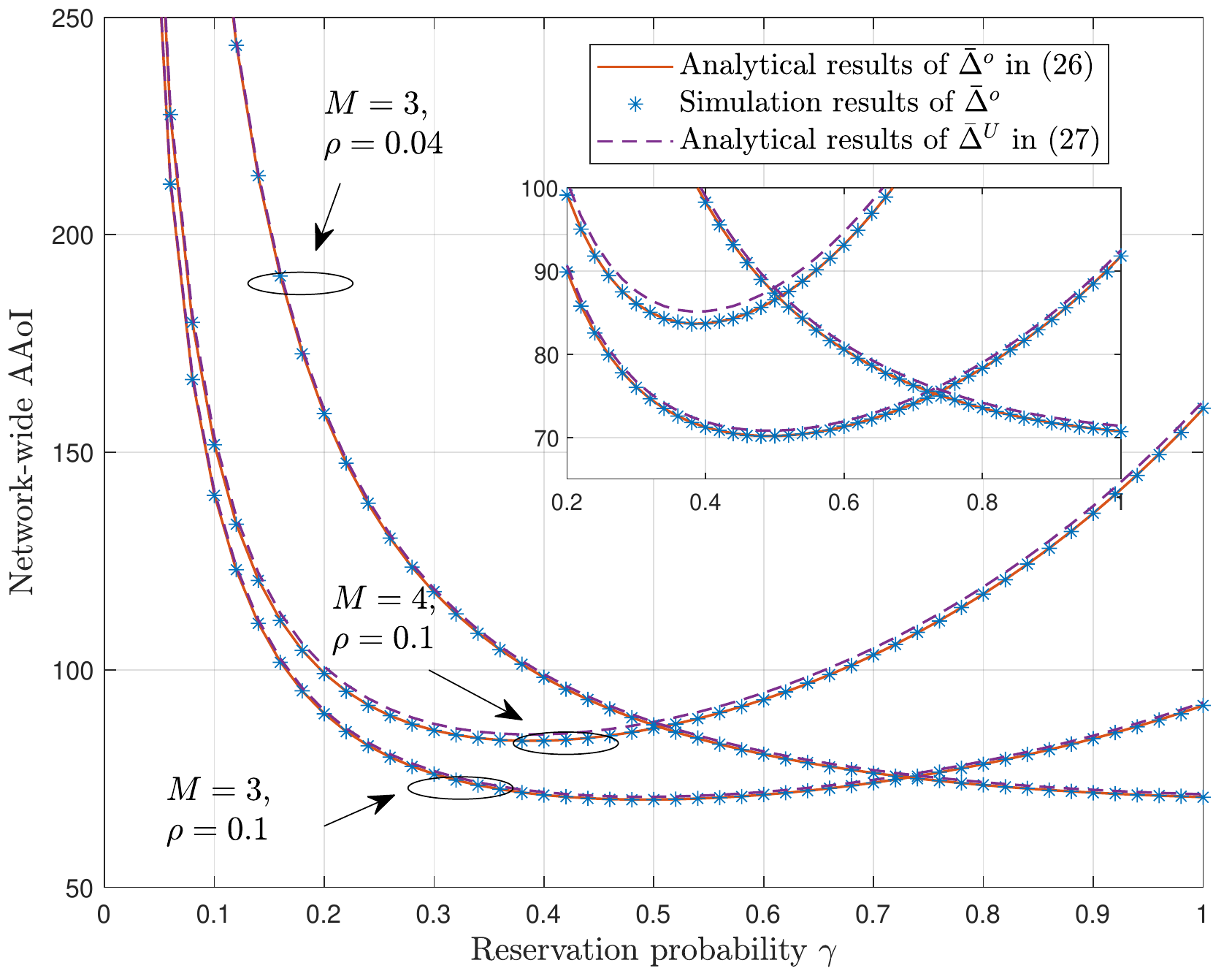}
		\caption{FSA-RD-One with $N=30$, $V=4$.}
		\label{fig3b}
	\end{subfigure}
 	\begin{subfigure}[b]{0.45\textwidth}
		\centering
  \includegraphics[width=\textwidth]{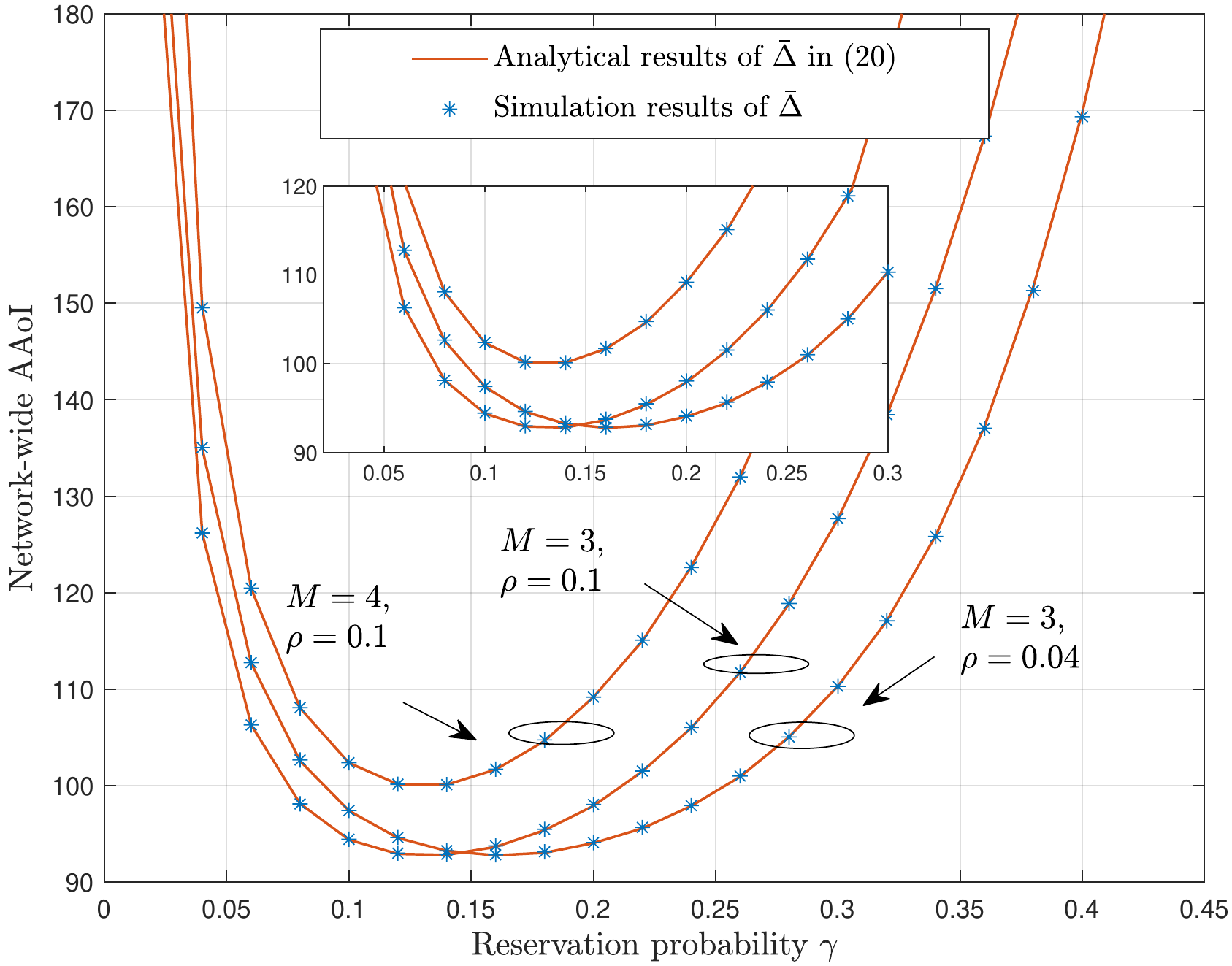}
		\caption{FSA-RD with $N=50$, $V=6$.}
		\label{fig3c}
	\end{subfigure}
 	\begin{subfigure}[b]{0.45\textwidth}
		\centering
  \includegraphics[width=\textwidth]{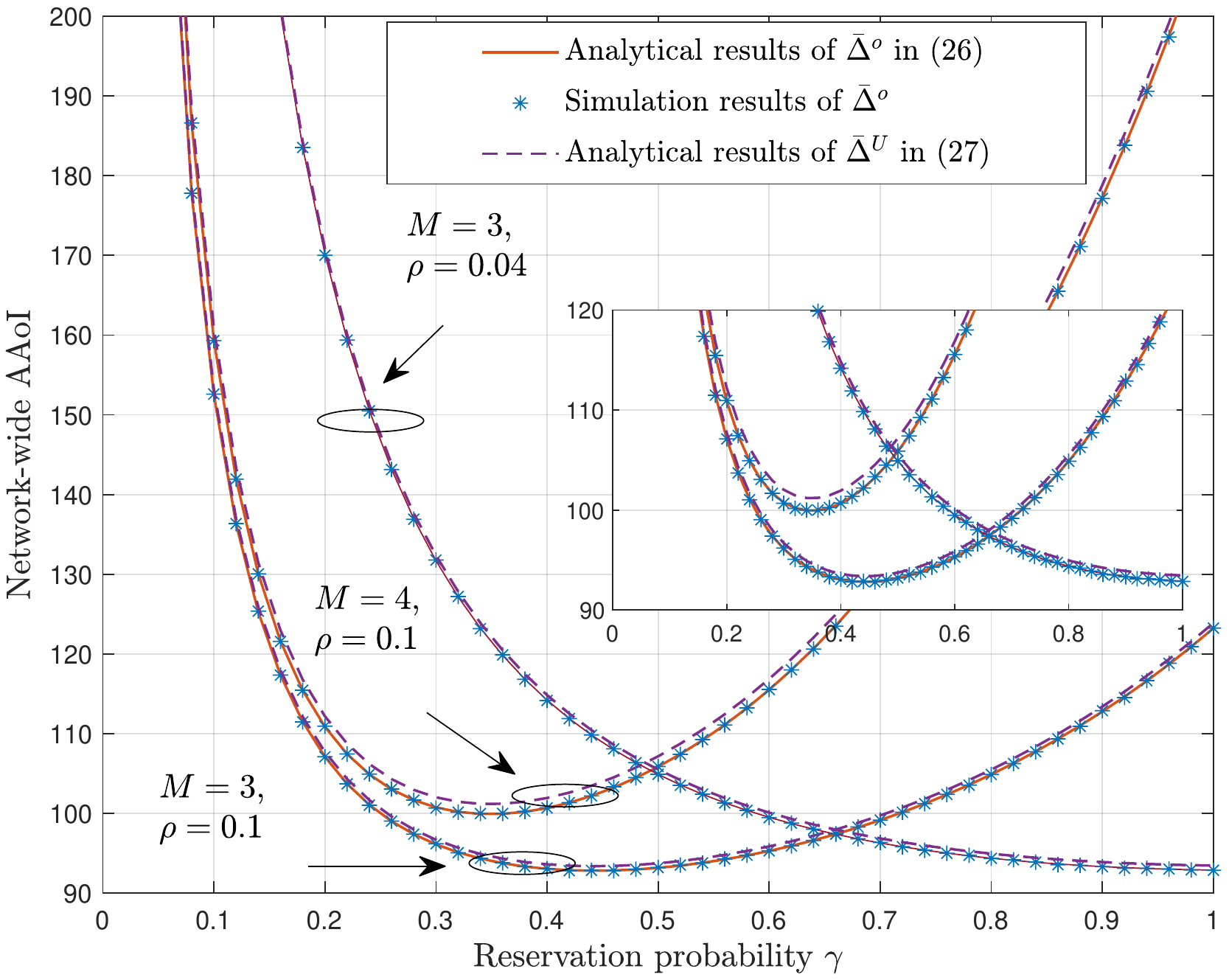}
		\caption{FSA-RD-One with $N=50$, $V=6$.}
		\label{fig3d}
	\end{subfigure}
	\caption{Network-wide AAoI versus reservation probability $\gamma$ for different frame length $M$ and status generation probability $\rho$.}
	\label{fig3*}
\end{figure}}
We first evaluate the derived analytical expression of AAoI for FSA-RD and FSA-RD-One by comparing them with Monte Carlo simulation results. {Figs. \ref{fig3a}, \ref{fig3c} and Figs. \ref{fig3b}, \ref{fig3d}} plot the AAoI curves versus reservation probability $\gamma$ of FSA-RD and FSA-RD-One, respectively, {considering multi-access systems with $N=30$, $V=4$ and $N=50$, $V=6$}. Each data point is obtained by averaging over $10^7$ time slots.

We can see from Fig. \ref{fig3*} that our analytical results coincide well with the corresponding simulation results, which validates our derivation. For FSA-RD-One, we also plot the upper bound of AAoI in \eqref{ageu} to verify the tightness.
We can also find that when the status arrival rate goes large, the optimal reservation probability should be neither too large nor too small for smaller AAoI. This is understandable, as a higher reservation probability is more likely to cause reservation failures (i.e., more reservation collisions), while a lower reservation probability makes users less likely to make a reservation. In both cases, fewer users transmit status updates in data slots, leading to a larger network-wide AAoI. 

Meanwhile, we observe that the optimal reservation probability of FSA-RD is smaller than that of FSA-RD-One. This confirms our analysis in Remark \ref{rem1} on the relationship between the optimal reservation probability of FSA-RD and that of FSA-RD-One. Furthermore, for FSA-RD-One, when the status generation probability is considerably small (e.g., $\rho=0.04$), the reservation probability should be as large as possible. The rationale is that when status updates are rarely generated, it is less likely to cause collision even when all users make reservations once they have a status update to transmit. In contrast, a small reservation probability may lead to the drop of rarely generated status updates, resulting in a larger AAoI.

In addition, comparing the performance of different frame lengths $M$, we observe that larger $M$ in some cases may cause performance degradation for both FSA-RD and FSA-RD-One. This is because when only a small number of users make successful reservations, after these users transmitting their status updates, the rest of the data slots in the frame will be wasted. 
\subsection{Evaluation of parameter design for FSA-RD-One}
\begin{figure}[!t]
	\centering
	\begin{subfigure}[b]{0.45\textwidth}
		\centering
		\includegraphics[width=\textwidth]{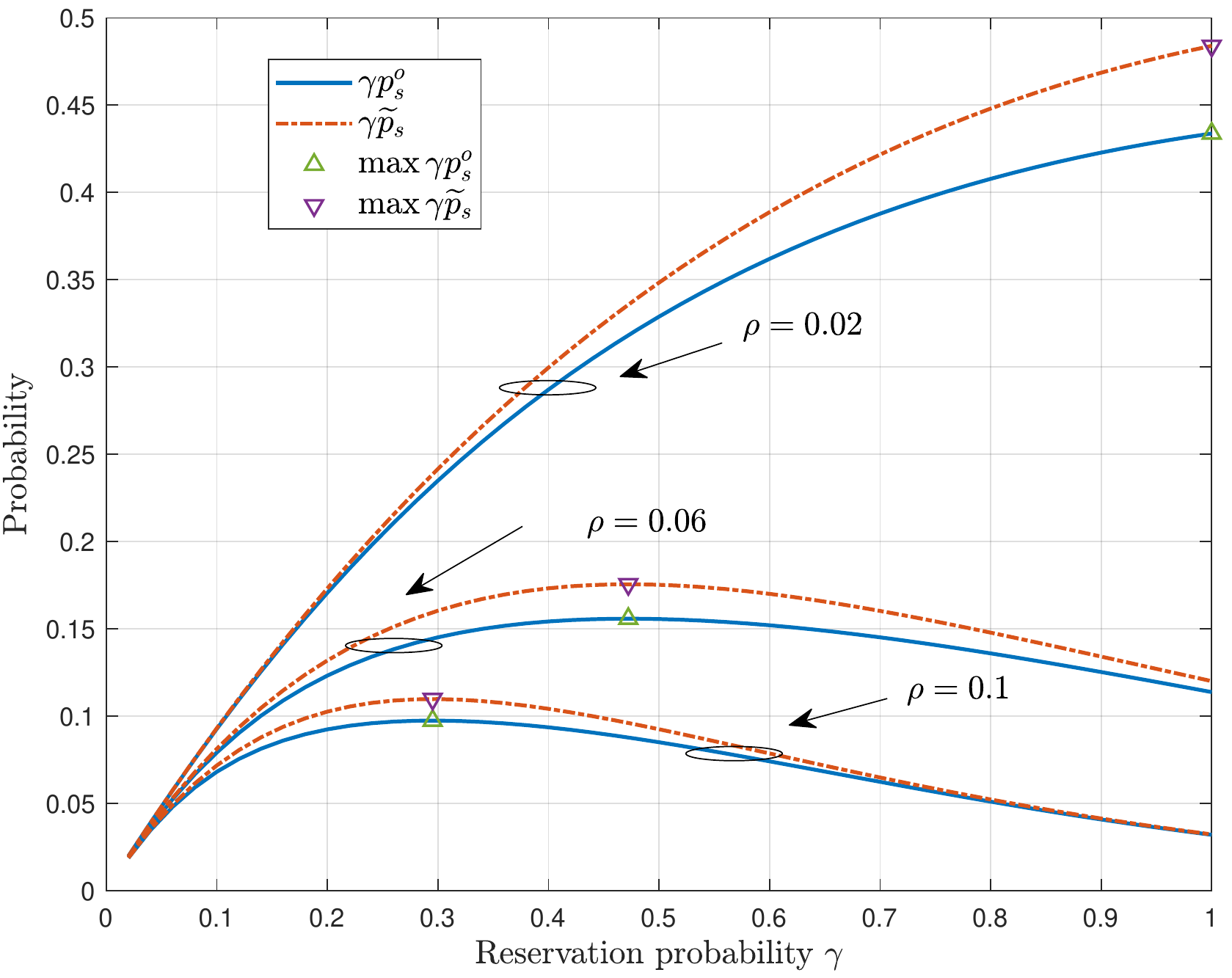}
		\caption{The relationship between $\widetilde{p}_s\gamma$ and $p_s^o\gamma$ with $M=3$ and $V=4$.}
		\label{fig4a}
	\end{subfigure}
	\begin{subfigure}[b]{0.45\textwidth}
		\centering
		\includegraphics[width=\textwidth]{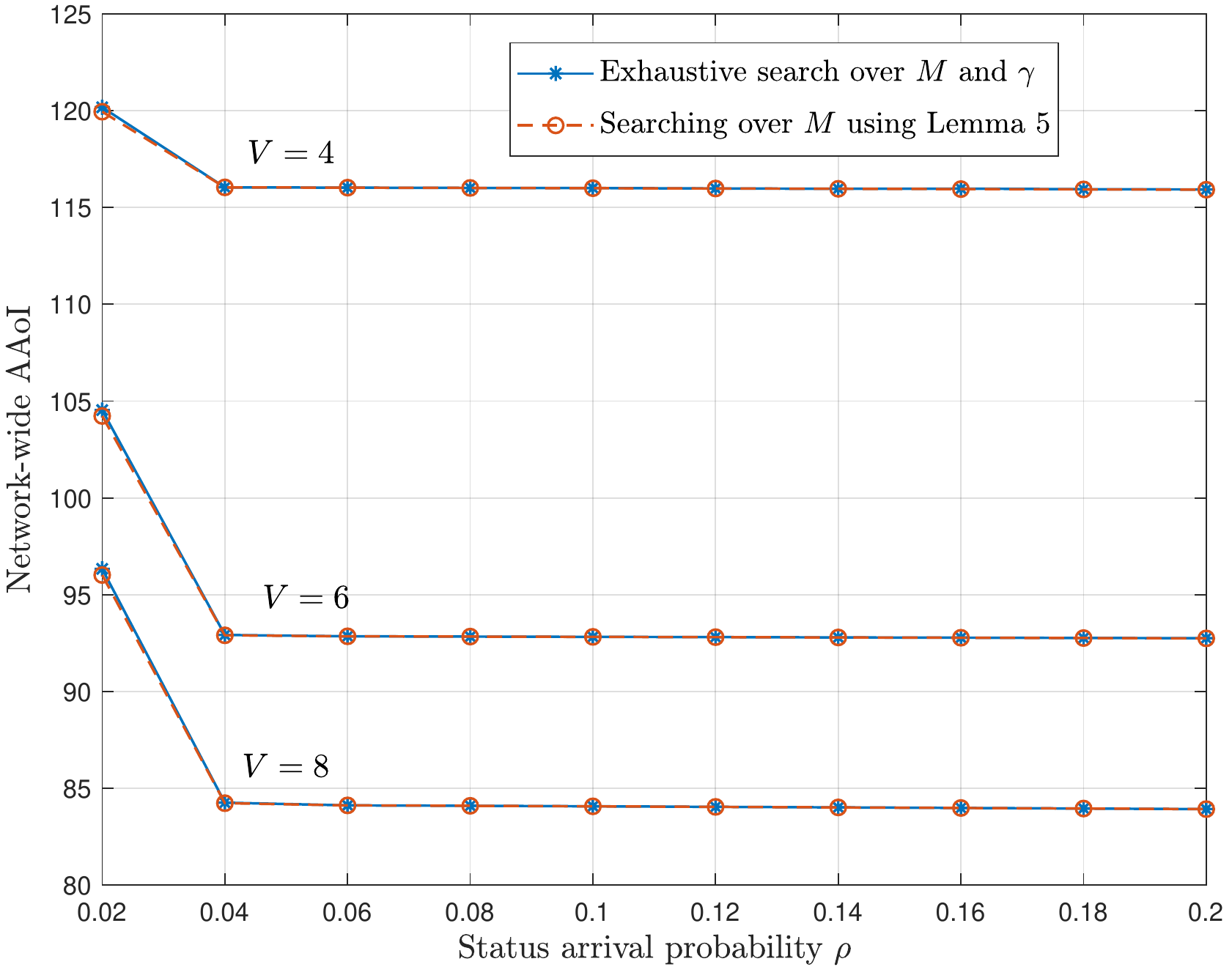}
		\caption{The comparison of optimized performance via exhaustive search and that via  using Lemma \ref{l5}.}
		\label{fig4b}
	\end{subfigure}
	\caption{Verification of the parameter design for FSA-RD-One considering a 50-user network.}
	\label{fig4}
\end{figure}
In Fig. \ref{fig4}, we evaluate our analysis for optimizing FSA-RD-One. We first evaluate the relationship between $p_s^o\gamma$ and $\widetilde{p}_s\gamma$ in Fig. \ref{fig4a}. Nearly the same trends of $p_s^o\gamma$ and $\widetilde{p}_s \gamma$ versus $\gamma$ in different cases validate the effectiveness of Lemma \ref{l5}. Moreover, by substituting $V$, $M$ and $\rho$ into the expression of $\gamma^*$ in Lemma \ref{l5}, we can obtain the maxima of $\widetilde{p}_s \gamma$. By comparing it with the maxima of $p_s^o \gamma$, they have nearly the same optimal point. After that, we compare the optimized performance obtained via exhaustive search over $M$ and $\gamma$ with that obtained by searching over $M$ with the help of Lemma \ref{l5}. Fig. \ref{fig4b} shows the closeness of the optimized performance obtained by these two methods, which illustrates the effectiveness of our optimization method. Moreover, compared with the two-dimensional exhaustive search where one of the parameters is continuous, using Lemma \ref{l5} makes parameter design much more simplified, where we get rid of searching over the continuous parameter $\gamma$ and focus on optimizing discrete variable $M\in[2, V+1]$. Furthermore, we can see that the larger the number of mini-slots (i.e., the value of $V$) in the reservation slot, the better the network-wide AAoI. This is natural as the larger the value of $V$, the higher the successful reservation probability.
\subsection{Performance comparison among FSA-RD, FSA-RD-One and slotted ALOHA}
In this subsection, we first conduct performance comparisons between FSA-RD and FSA-RD-One over different system setups in Fig. \ref{fig3}. Specifically, for FSA-RD, we adopt an exhaustive search over frame size $M$ and reservation probability $\gamma$ to obtain the optimal performance; for FSA-RD-One, we make use of Lemma \ref{l5} to obtain the near-optimal performance for each $M$ and compare the performance of different values of $M$. We find that when the network status arrival rate (i.e., $\rho N$) is small, the AAoI of the optimized FSA-RD is better than that of FSA-RD-One. As $\rho N$ increases, the AAoI of these two transmission schemes becomes similar. This is because when the status arrival rate is small, the rarely arrived status update has at most one reservation opportunity regardless of its reservation result, leading to worse AAoI for FSA-RD-One. When the status arrival rate becomes large, due to the status preemption during the transmission process, even when multiple reservation attempts are allowed for FSA-RD, status updates for transmission that are not received by the AP will be replaced by the newly arrived status update rather than transmitted to the AP. In this scenario, FSA-RD will behave like FSA-RD-One leading to similar AAoI. This phenomenon can also be observed in Table \ref{table1}.
\begin{figure}[!t]
	\centering \scalebox{0.5}{\includegraphics{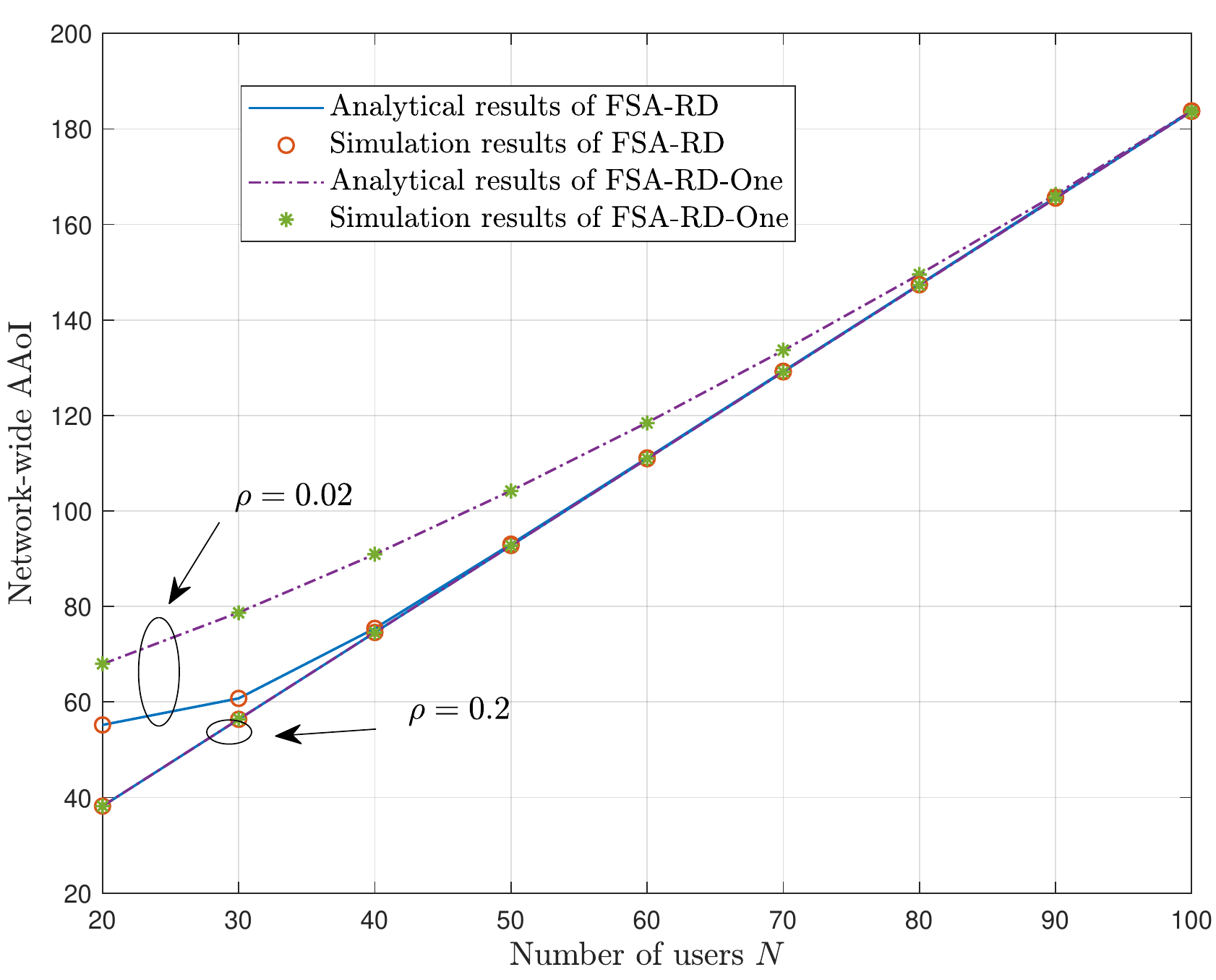}}
	\caption{Optimized network-wide AAoI versus the number of users $N$ with $V=6$.}
	\label{fig3}
\end{figure}

\begin{table}[!http]
	\begin{subtable}[t]{\textwidth}
		\centering
 \resizebox{\textwidth}{!}{%
\begin{tabular}{|c|*{15}{c|}}
\hline
\multirow{2}{*}{} & \multicolumn{3}{c|}{$\rho=0.01$} &\multicolumn{3}{c|}{$\rho=0.02$}&\multicolumn{3}{c|}{$\rho=0.04$}&\multicolumn{3}{c|}{$\rho=0.08$}&\multicolumn{3}{c|}{$\rho=0.1$}\\
\cline{2-16}
& $\gamma^*$ & $M^*$ & AAoI & $\gamma^*$ & $M^*$ & AAoI& $\gamma^*$ & $M^*$ & AAoI& $\gamma^*$ & $M^*$ & AAoI& $\gamma^*$ & $M^*$ & AAoI\\
\hline
FSA-RD, $V=4$ & 0.82 & 2 & \textbf{105.55} & 0.38 & 2 & \textbf{72.38}& 0.20 & 3 & \textbf{70.25}& 0.16 & 3 & \textbf{70.16}&0.15 & 3 & \textbf{70.15}\\
\hline
FSA-RD, $V=6$ & 1 & 2 & \textbf{104.37} & 0.85 & 3 & \textbf{60.75}& 0.35 & 3 & \textbf{56.53}& 0.25 & 3 & \textbf{56.45}&0.24 & 3 & \textbf{56.45}\\
\hline
FSA-RD, $V=8$ & 1 & 3 & \textbf{104.16} & 1 & 3 & \textbf{57.84}& 0.50 & 3 & \textbf{51.38}& 0.34 & 3 & \textbf{51.32}&0.32 & 3 & \textbf{52.30}\\
\hline
FSA-RD-One, $V=4$ & 1 & 3 & 131.16 & 1 & 3 & 86.46& 1 & 3 & \textbf{70.74}& 0.6025 & 3 & \textbf{70.18}&0.4920 & 3 & \textbf{70.16}\\
\hline
FSA-RD-One, $V=6$ & 1 & 3 & 124.06& 1 & 3 &\textbf{78.74}& 1 & 3 & \textbf{60.42}& 0.9037 & 3 & \textbf{56.47}&0.7380 & 3 & \textbf{56.46}\\
\hline
FSA-RD-One, $V=8$ & 1 & 3 & 120.82& 1 & 4&\textbf{74.55}& 1 & 4 & \textbf{55.67}& 0.9403 & 4 & \textbf{51.37}&0.9840 & 3 & \textbf{51.32}\\
\hline
slotted ALOHA &\diagbox{}{}  & \diagbox{}{}& 110.14& \diagbox{}{}& \diagbox{}{}&82.55& \diagbox{}{} & \diagbox{}{} & 81.30& \diagbox{}{} & \diagbox{}{} & 80.22&\diagbox{}{} & \diagbox{}{} & 80.12\\
\hline
\end{tabular}}
		\caption{Network-wide AAoI versus packet generation probability $\rho$ with $N=30$.}
		\label{tab:week1}
	\end{subtable} \\
	\begin{subtable}[t]{\textwidth}
		\centering
\resizebox{\textwidth}{!}{%
		\begin{tabular}{|c|*{12}{c|} }
		\hline
  \multirow{2}{*}{} & \multicolumn{3}{c|}{$N=10$} &\multicolumn{3}{c|}{$N=20$}&\multicolumn{3}{c|}{$N=40$}&\multicolumn{3}{c|}{$N=50$}\\
\cline{2-13}
& $\gamma^*$ & $M^*$ & AAoI & $\gamma^*$ & $M^*$ & AAoI& $\gamma^*$ & $M^*$ & AAoI& $\gamma^*$ & $M^*$ & AAoI\\
\hline
		FSA-RD, $V=4$ & 1 & 2 & \textbf{30.58} & 0.4 & 3 &\textbf{47.71} & 0.13 & 3 &\textbf{93.14}&0.10 & 3 &\textbf{116.02}\\
		\hline
		FSA-RD, $V=6$ & 1 & 3 & \textbf{29.77}& 0.77 &3 &  \textbf{38.89}& 0.22 &3 &\textbf{74.67}& 0.16 & 3 &\textbf{92.84}\\
  \hline
		FSA-RD, $V=8$ & 1 & 3 & \textbf{29.45}& 1 &3 &  \textbf{35.78}& 0.51 &3 &\textbf{67.73}& 0.22 & 3 &\textbf{84.12}\\
		\hline
		FSA-RD-One, $V=4$ &1 & 3 & 37.40& 1 & 3 &  \textbf{52.12}& 0.8676 & 3 & \textbf{93.12} & 0.6941 & 3 &\textbf{116.04}\\
		\hline
		FSA-RD-One, $V=6$ & 1 & 3 & 35.12&  1 & 3 & \textbf{46.63}& 1 & 3 & \textbf{75.89}& 1 & 3 &\textbf{92.90}\\
  \hline
		FSA-RD-One, $V=8$ & 1 & 3 & 34.09&  1 & 4 & \textbf{43.89}& 1 & 4 & \textbf{69.19}& 1 & 4 &\textbf{84.23}\\
		\hline
		slotted ALOHA & \diagbox{}{} & \diagbox{}{} & 31.63& \diagbox{}{} & \diagbox{}{} & 53.72 &\diagbox{}{} & \diagbox{}{} & 107.66 & \diagbox{}{} & \diagbox{}{}& 136.97\\
		\hline	
		\end{tabular}}
           
		\caption{Network-wide AAoI versus total number of users $N$ with $\rho=0.04$.}
		\label{tab:week2}
	\end{subtable}
	\caption{Performance comparison among optimized FSA-RD, FSA-RD-One and slotted ALOHA.}
	\label{table1}
\end{table}

In Table \ref{table1}, we compare the optimized AAoI of FSA-RD and FSA-RD-One, with that of slotted ALOHA. Specifically, given the network setups, {we use the same method as in Fig. \ref{fig3} to obtain the values of the optimized reservation probability $\gamma^*$ and frame length $M^*$ as well as the corresponding AAoI for FSA-RD and the near-optimal AAoI for FSA-RD-One.} As for slotted ALOHA, we simulate the age evolution with different values of transmission probability to find its optimal AAoI. We can see that FSA-RD always outperforms slotted ALOHA and the performance gap between them enlarges as the $\rho N$ becomes large. For FSA-RD-One, when $\rho N$ is small, the slotted ALOHA has a smaller AAoI. As $\rho N$ increases through either increasing $N$ or $\rho$, the optimal performance of FSA-RD-One improves and becomes better than that of slotted ALOHA with substantial AAoI reduction, 
approaching the performance of FSA-RD. This observation of FSA-RD-One could be jointly caused by the one-slot reservation overhead and the frame-based transmission. Specifically, under FSA-RD-One, the reservation slot in each frame will be used for reservation and no status update can be transmitted within this slot, while slotted ALOHA makes use of every slot for status update. In addition, each status update can be transmitted at most once in FSA-RD-One due to the constraint of one reservation attempt per status update. This will lead to the drops of rarely generated status updates when the status generation rate is small. By contrast, slotted ALOHA allows retransmission for each status update as long as it is not preempted by a newly generated status update. These jointly yield that FSA-RD-One is inferior to slotted ALOHA when the network-wide status arrival rate is small.

Furthermore, upon observing the optimized frame size and reservation probability, we can deduce that for any system setup, the optimal reservation probability is non-decreasing in the number of mini-slots $V$, when the frame size is fixed. That is, a larger value of $V$ corresponds to a higher value of the optimized $\gamma^*$, given the frame size $M^*$, for both FSA-RD and FSA-RD-One. Moreover, the optimized reservation probability is non-increasing in the number of users $N$ and the status generation probability $\rho$ given the values of $V$ and $M^*$.  This result is consistent with the monotonicity of $\gamma^*$ stated in Lemma \ref{l5}. As for $M^*$, in most cases, it increases as the values of $V$, $N$, and $\rho$ increase for both FSA-RD and FSA-RD-One. However, we observed an exception in FSA-RD-One, the optimized frame size for $\rho=0.1$, $V=8$ and $N=30$ is $3$, mathematically, $M^{*}(\rho=0.1,V=8,N=30)=3$, while $M^{*}(\rho=0.08,V=8,N=30)=4$. This situation could be attributed to the combined effect of reservation probability and frame size for a better AAoI. Specifically, although the optimized frame size $M^*$ for $\rho=0.1$ is smaller than that for $\rho=0.08$, the corresponding optimized reservation probability $\gamma^*$ for the former case is larger than that for the latter case. Besides, the approximation made  to derive the near-optimal reservation probability may also be a possible reason for this exception.
\section{Conclusions}
{We investigated the average age of information (AAoI) of Frame Slotted ALOHA with Reservation and Data slots (FSA-RD), and its simplified version, FSA-RD with one reservation attempt per status update (FSA-RD-One). We derived an analytical expression for the AAoI of FSA-RD. We also attained closed-form expressions for the AAoI of FSA-RD-One and its near-optimal reservation probability. The correctness of the AAoI analysis of the two schemes was confirmed by numerical results. Simulation results showed that the optimized FSA-RD always outperforms the optimized slotted ALOHA. When the status arrival rate of the network becomes large, FSA-RD and FSA-RD-One, with optimized reservation probability and frame size, will achieve similar AAoI performance, which is substantially lower than that of the optimized slotted ALOHA. Future work will include the investigation of the AAoI of FSA-RD with variable frame length and other contention schemes for making reservations.} 

\begin{appendices}
\section{Proof of Lemma \ref{l1}}
\label{pnl1}
	Recall that $X_{i}=M-\alpha_{i-1}+Y_{i}-M+\alpha_{i}$, we have $\mathbb{E}[X_i^2]=\mathbb{E}[\alpha_{i}^2]+\mathbb{E}[\alpha_{i-1}^2]-2\mathbb{E}[\alpha_{i}\alpha_{i-1}]+\mathbb{E}[Y_i^2]+\mathbb{E}[2(\alpha_{i}-\alpha_{i-1})Y_i]$. As $\alpha_{i}$ and $\alpha_{i-1}$ are i.i.d., and both of them are independent of $Y_i$, we have $\mathbb{E}[2(\alpha_{i}-\alpha_{i-1})Y_i]=0$ and $\mathbb{E}[\alpha_{i}^2]+\mathbb{E}[\alpha_{i-1}^2]-2\mathbb{E}[\alpha_{i}\alpha_{i-1}]=2\mathrm{Var}(\alpha)$. This completes the proof.

\section{Proof of Lemma \ref{l2}}
\label{pnl2}
	Recall that $X_{i}=M-\alpha_{i-1}+Y_{i}-M+\alpha_{i}$ and $S_{i-1}=l_{i-1}+Z_{i-1}-M+\alpha_{i-1}$. Thus, $\mathbb{E}[S_{i-1}X_i]=\mathbb{E}[(l_{i-1}+Z_{i-1}-M+\alpha_{i-1})(M-\alpha_{i-1}+Y_{i}-M+\alpha_{i})]\overset{a}{=}\mathbb{E}[S_{i-1}]\mathbb{E}[Y_i]+\mathbb{E}[(l_{i-1}+Z_{i-1}-M+\alpha_{i-1})(\alpha_i-\alpha_{i-1})]\overset{b}{=}\mathbb{E}[S_{i-1}]\mathbb{E}[Y_i]+\mathbb{E}[l_{i-1}+Z_{i-1}-M]\mathbb{E}[\alpha_{i}-\alpha_{i-1}]+\mathbb{E}[\alpha_{i-1}(\alpha_{i}-\alpha_{i-1})]$, where $\overset{a}{=}$ is because of the independence among $l_{i-1}$, $\alpha_{i-1}$, $Z_{i-1}$ and $Y_i$; and $\overset{b}{=}$ is because of the independence among $\alpha_{i}$, $\alpha_{i-1}$, $l_{i-1}$ and $Z_{i-1}$. As $\alpha_{i}$ and $\alpha_{i-1}$ are i.i.d., we have $\mathbb{E}[\alpha_{i}-\alpha_{i-1}]=0$ and $\mathbb{E}[\alpha_{i-1}(\alpha_{i}-\alpha_{i-1})]=-\mathrm{Var}(\alpha)$. This completes the proof.

	\section{Proof of Lemma \ref{l3}}
	\label{pl1}
	As we consider a symmetric network, $p_s$ is the same for all users. We thus consider the $p_s$ for arbitrary user $n$. Note that $p_s$ depends on the number of users excluding user $n$ that decide to make reservations in a certain frame, denoted by $N_r$. We have $N_r\in\{0,1,...,N_g\}$, where $N_g$ denotes the number of active users excluding user $n$ with a status update packet to transmit, and thus $N_g\in\{0,1,...,N-1\}$. According to the steady state distribution of the number of active users $\boldsymbol{\pi}$, the probability that there are $n_1+1$ active users including user $n$ is given by $\psi_{n_1}=\pi_{n_1+1}\frac{n_1+1}{N}$, where $\frac{n_1+1}{N}=\frac{C_{N-1}^{n_1}}{C_N^{n_1+1}}$ is the probability that $n_1+1$ active user selections including user $n$. We then normalize $\psi_{n_1}$ to obtain the distribution for $N_g$. That is, $\mathrm{Pr}\left\{N_g=n_1\right\}=\frac{\pi_{n_1+1}(n_1+1)}{\sum_{j=0}^{N-1}\pi_{j+1}(j+1)}$. According to the reservation scheme, the conditional PMF of $N_r$ given $N_g=n_1$ can be expressed as $\mathrm{Pr}\left(N_r=n_2|N_g=n_1\right)=\mathrm{B}_{n_1}^{n_2}(\gamma)$. Together with user $n$, there will be $N_r+1$ users randomly selecting one of the $V$ min-slots to send a short reservation packet. 
 
 The number of min-slots that are reserved by one single user, denoted by $N_s$, depends on the total number of reservation users.  According to \cite[Eq. 6]{szpankowski1983analysis} and \eqref{eqr1}, we have $\mathrm{Pr}\left(N_s=n_3|N_r=n_2\right)=\mathrm{R}^{n_3}_{n_2+1}(V)$ denoting the probability that $n_2+1$ users successfully reserve $n_3$ mini-slots. Due to the identical reservation scheme of each user and the limited $M-1$ data slots, the successful reservation probabilities of the $n_2+1$ users are the same. As such, for any of the $n_2+1$ users, the probability of successful status update is $\frac{\min\{n_3,M-1\}}{n_2+1}$, i.e., $\mathrm{Pr}\left\{I_a^n=1|N_s=n_3,N_r=n_2\right\}=\frac{\min\{n_3,M-1\}}{n_2+1}$. Based on all the above analysis and the law of total probability, we arrive at the expression of $p_s$ given in Lemma \ref{l3}. This completes the proof.
 \section{Proof of Lemma \ref{lem1}}
 \label{pl4}
 Note that $p_s^o$ depends on the number of users excluding user $n$ that decide to make reservations in a certain frame, denoted by $N_r$. We have $N_r\in\{0,1,...,N_g\}$, where $N_g$ denotes the number of users excluding user $n$ that have a status update packet to transmit, and thus $N_g\in\{0,1,...,N-1\}$. According to the status generation model and reservation scheme, the PMF of $N_g$ and the conditional PMF of $N_r$ given $N_g=n_1$ can be expressed as  $\mathrm{Pr}\left\{N_g=n_1\right\}=\mathrm{B}_{n_1}^{n_2}(1-(1-\rho)^M)$ and $\mathrm{Pr}\left(N_r=n_2|N_g=n_1\right)=\mathrm{B}_{n_1}^{n_2}(\gamma)$, respectively. Together with user $n$, there will be $N_r+1$ users randomly selecting one of the $V$ min-slots to send a short reservation packet. The number of min-slots that are reserved by one single user, denoted by $N_s$, depends on the total number of reservation users. According to \cite[Eq. 6]{szpankowski1983analysis}, we have $\mathrm{Pr}\left(N_s=n_3|N_r=n_2\right)=\mathrm{R}^{n_3}_{n_2+1}(V)$ denoting the probability that $n_2+1$ users successfully reserve $n_3$ mini-slots. Due to the identical reservation scheme of each user and the limitation of $M-1$ data slots, the successful reservation probabilities of the $n_2+1$ users are the same. For any of the $n_2+1$ users, the probability of successful status update is $\frac{\min\{n_3,M-1\}}{n_2+1}$, i.e., $\mathrm{Pr}\left\{I_a^n=1|N_s=n_3,N_r=n_2\right\}=\frac{\min\{n_3,M-1\}}{n_2+1}$. Based on all the above analysis and the law of total probability, we arrive at the expression of $p_s^o$ given in Lemma \ref{lem1}. This completes the proof.

\section{Proof of Corollary \ref{co6}}
\label{pc4}
    As most of the proof is similar to that in Lemma \ref{lem1}, we offer a sketch of it here. Specifically, $\widetilde{p}_s$ depends on the number of users making reservations excluding user $n$ (i.e., $n_2$) in each frame, and the number of users making reservations depends on the number of users with a status update to transmit (i.e., active users $n_1$). Given $n_2$ users making reservations, the successful transmission probability of user $n$ making a reservation is $V\frac{1}{V}\left(1-\frac{1}{V}\right)^{n_2}$, which is the product of $V$ mini-slot choices,  the probability of each choice, $1/V$, and the probability that the rest $n_2$ users do not select the chosen mini-slot,  $\left(1-\frac{1}{V}\right)^{n_2}$. After applying the law of total probability and some manipulations, we arrive at the expression of $\widetilde{p}_s$. This completes the proof.
 
\section{Proof of Lemma \ref{l5}}
	\label{pl5}
	Based on the expression of $\widetilde{p}_s$ given in Corollary \ref{co6}, 
 we calculate the first order derivative of $\widetilde{p}_s\gamma$ and get \[-\frac{V\cdot\left(N\left(1-(1-\rho)^M\right){\gamma}-V\right)\left(1-\frac{\left(1-(1-\rho)^M\right){\gamma}}{V}\right)^N}{\left(\left(1-(1-\rho)^M\right){\gamma}-V\right)^2}.\] We then calculate the zero-root of the first order derivative and analyze the sign to determine the monotonicity to obtain $\gamma^*$. This completes the proof.
\end{appendices}


%

\ifCLASSOPTIONcaptionsoff
  \newpage
\fi



%
\bibliography{ref}
\bibliographystyle{IEEEtran}

%




\end{document}